\documentclass[journal]{IEEEtran} 
\usepackage{cite}  
\usepackage{times}
\usepackage{amsmath,amsfonts,amssymb,mathtools}

\usepackage{epsfig,amstext,xspace}

\usepackage{amsthm}

\usepackage{graphicx} \usepackage{caption} \usepackage{subcaption}

\usepackage{algpseudocode}
\usepackage{multirow}
\usepackage[table]{xcolor}

\newtheorem{innercustomthm}{Theorem}
\newenvironment{customthm}[1]
{\renewcommand\theinnercustomthm{#1}\innercustomthm}
{\endinnercustomthm}

\newtheorem{innercustomlem}{Lemma}
\newenvironment{customlem}[1]
{\renewcommand\theinnercustomlem{#1}\innercustomlem}
{\endinnercustomlem}

\usepackage{algorithm}
\usepackage{array}
\newtheorem{theorem}{Theorem}

\newtheorem{lemma}[theorem]{Lemma}

\newtheorem{claim}{Claim}

\newtheorem{definition}{Definition}

{\theoremstyle{remark} \newtheorem{remark}{\bf Remark}}
{\theoremstyle{definition}

}

\def\CC{\mathbb C}
\def\RR{\mathbb R}

\def\cF{\mathcal F}
\def\cG{\mathcal G}
\def\cI{\mathcal I}

\def\cL{\mathcal L}

\def\cN{\mathcal N}

\newcommand{\raf}[1]{(\ref{#1})}

\newcommand{\cV}{\ensuremath{\mathcal{V}}}
\newcommand{\cE}{\ensuremath{\mathcal{E}} }

\newcommand{\OPT}{\ensuremath{\textsc{Opt}}}

\newcommand{\polylog}{\operatorname{polylog}}

\newcommand{\re}{\ensuremath{\mathrm{Re}}}
\newcommand{\im}{\ensuremath{\mathrm{Im}}}

\title{Optimal Power Flow with Inelastic Demands for Demand Response in Radial Distribution Networks\thanks{Extended version of the journal paper appears in IEEE Transactions on Control of Network Systems. DOI: http://dx.doi.org/10.1109/TCNS.2016.2622362}}

\author{
	Majid Khonji, Chi-Kin Chau, and Khaled Elbassioni
	\thanks{M. Khonji is  with Dubai Electricity and Water Authority (DEWA), Dubai, UAE (e-mail: majid.khonji@dewa.gov.ae). C.-K. Chau, and K. Elbassioni are with the Dept. of EECS at Masdar Institute of Science and Technology, UAE (e-mail: \{ckchau, kelbassioni\}@masdar.ac.ae).}
}
\newif\ifsupplementary
\supplementarytrue
\begin{document}

\maketitle

\begin{abstract}
The classical optimal power flow problem optimizes the power flow in a power network considering the associated flow and operating constraints. In this paper, we investigate optimal power flow in the context of utility-maximizing demand response management in distribution networks, in which customers' demands are satisfied subject to the operating constraints of voltage and transmission power capacity. The prior results concern only elastic demands that can be partially satisfied, whereas power demands in practice can be inelastic with binary control decisions, which gives rise to a mixed integer programming problem. We shed light on the hardness and approximability by polynomial-time algorithms for optimal power flow problem with inelastic demands. We show that this problem is inapproximable for general power network topology with upper and lower bounds of nodal voltage. Then, we propose an efficient algorithm for a relaxed problem in radial networks with bounded transmission power loss and upper bound of nodal voltage. We derive an approximation ratio between the proposed algorithm and the exact optimal solution. Simulations show that the proposed algorithm can produce close-to-optimal solutions in practice.
\end{abstract}

\pagenumbering{arabic}
\section{Introduction}

Electric power network is a distinctive networked system, because the flows in such a network obey certain non-linear physics laws, unlike other networked systems. Despite being a century-old system, the control and optimization of electric power networks is still a challenging fundamental problem that baffles electrical power engineers.

The optimal power flow (OPF) problem is an optimization problem that minimizes a certain cost (e.g., power loss) subject to Kirchhoff's laws of electric flows, and several operating constraints of nodal voltage and transmission power capacity. A critical challenge in optimizing electric power networks is the presence of complex-valued quantities (for modeling periodic varying properties such as voltage and current) and quadratic constraints (for describing the behavior of power flows) that give rise to a non-convex optimization problem. Recently, there have been a number of breakthroughs \cite{gan2015exact,low2014convex1,low2014convex2} in tackling OPF by convex relaxations, which are shown to attain the exact optimal solutions under certain mild conditions. 

Demand response management is a critical mechanism to balance power demand and supply. In this work, we consider optimal power flow in the context of utility-maximizing demand response management in distribution networks, in which customers' demands are satisfied subject to operating constraints of nodal voltage and transmission power capacity. While this problem has been considered in previous papers (e.g., \cite{li2012optimization}), the prior results concern only elastic demands that can be partially satisfied, whereas some power demands in practice can be inelastic with binary control decisions (e.g., appliances that can be either switched on or off). We formulate a mixed integer programming problem to model utility-maximizing demand response management with inelastic and elastic demands. We shed light on the hardness and approximability by polynomial-time algorithms for this problem.

This work unifies three separate strands of work. First, it is related to the optimal power flow for demand response management with elastic demands \cite{li2012optimization}. Second, the combinatorial power allocation with inelastic demands has been studied in the single-link case (known as {\em complex-demand knapsack problem} \cite{CKM14}). We extend the combinatorial power allocation problem to general networks with power flows. Third, allocation of inelastic demands of real-valued commodity is known as {\em unsplittable flow problem} \cite{chekuri2009unsplittable,anagnostopoulos2014mazing}, which is an important topic in theoretical computer science. This work applies several ideas from solving the unsplittable flow problem to the flow problem of complex-valued commodity (i.e., electric power).  But our results generalizing those of complex-demand knapsack problem and unsplittable flow problem are non-trivial, because of the significant challenges in tackling complex-valued flow subject to the operating constraints.

The contributions of this paper are summarized as follows: 
\begin{enumerate}

\item {\em Hardness}: We show that the optimal power flow for demand response management with inelastic demands is inapproximable (even allowing constraint violation), when considering upper and lower bounds of nodal voltage, or transmission power capacity constraints in a general (cyclic) electric network.

\item {\em Approximability}: We propose an efficient approximation algorithm for a relaxed problem in radial networks (i.e., trees) with bounded transmission power loss and upper bound of nodal voltage. We derive an approximation ratio bounding the gap between the solution of the proposed algorithm and the exact optimal solution (which is computationally hard to obtain). 

\item {\em Mixed Demands}: We provide an extension of our approximation algorithm that can handle a mix of both elastic and inelastic demands.

\item {\em Evaluations}: We perform extensive simulations on a test electric network to evaluate the practical performance of our algorithms. The proposed algorithm is observed to produce close-to-optimal solutions in practice, while being faster by orders of magnitude compared to an exact mixed integer programming solver. 

\end{enumerate}

\section{Related work}

\subsection{Optimal Power Flow Problem}

In general power networks, the optimal power flow (OPF) problem is characterized by constraints defined by two models: Bus Injection Model (BIM)  and Branch Flow Model (BFM) (also called DistFlow model)\cite{low2014convex1,low2014convex2}. Variables in BIM (i.e., voltage and power) are assigned for every bus (or node). On the other hand, variables in BFM are assigned for every branch (or edge). BFM was first proposed in \cite{baran1989placement, baran1989sizing}. It has been shown by \cite{equal} that both models are in fact equivalent, in a sense that there is a bijection map between solutions in both models.
In radial (or tree) topologies one can reduce the number of optimization variables in BFM. Given a solution, one can recover all omitted angles uniquely in tree topologies through a polynomial time procedure described in \cite{farivar2013branch}. 
The  power  flow  equations,  in  either models,  are  nonlinear  and  the  solution  sets  are  non-convex, therefore hard to compute.
One way to solve OPF is to relax the feasible region to become convex. 
In the branch flow model, a second order cone (SOC) relaxation is shown in \cite{gan2015exact} to be exact under mild conditions in  tree topologies. Efficient algorithms exist for SOC programing.  
Other works obtained relaxations for BIM as well for tree topology that are also exact under some conditions \cite{low2014convex2}.


\subsection{Complex-demand Knapsack Problem}

On the other hand, there are several recent studies on demand response with inelastic demands.
For a single-link case, demand response with inelastic demands has been studied as the complex-demand knapsack problem {\sc (CKP)} and its application to power demand allocation was highlighted by \cite{YC13CKP}. Let  $\theta$ be the maximum angle between any complex valued demands. \cite{YC13CKP} obtained a $\frac{1}{2}$-approximation for the case where $0 \le \theta \le \frac{\pi}{2}$.  \cite{woeginger2000does} (also \cite{YC13CKP}) proved that no fully polynomial-time approximation scheme (FPTAS) exist. Recently, \cite{CKM14, CKM15} provided a polynomial-time approximation scheme (PTAS), and a bi-criteria FPTAS (allowing constraint violation) for $\frac{\pi}{2} < \theta < \pi - \varepsilon$, which closes the approximation gap. Also, \cite{KKCEZ16} provides a greedy efficient algorithm for solving CKP, and preliminary hardness result appears in \cite{KCE14}.

\subsection{Unsplittable Flow Problem}
When the demands are real-valued, our problem is related to the unsplittable flow problem (UPF). In UPF, each demand is associated with an arbitrary path from a source to a sink, while in our problem all demands share a single source (or sink). 
A special case of UPF is when all demands and edge capacities are uniform is the classical maximum edge-disjoint path problem (MEDP)
\cite{KR72}.   In directed graphs, the best  known approximation is $O( \min\{ \sqrt{m},n^{\frac{2}{3}} \log^{\frac{1}{3}} n \})$ \cite{KM96,VAV04}, while it is NP-Hard to approximate within $\Omega(n^{\frac{1}{2}-\epsilon})$ \cite{guruswami2003near}, where $n$ and $m$ are the number of nodes and edges respectively. 
\cite{AYR01} shows that UFP in directed graphs is $\Omega(n^{1-\epsilon})$-hard unless ${\rm P} = {\rm NP}$. In undirected graphs, there is an $O(\sqrt n )$-approximation \cite{CKS06}, and the best known hardness result is $\Omega(\log^{\frac{1}{2}-\epsilon} n)$ assuming ${\rm NP} \not \subseteq {\rm ZPTIME}(n^{O(\polylog(n))})$ \cite{AMC05}.  These hardness results suggest that the problem is difficult to solve in general graphs. For tree topology, the problem is APX-Hard (i.e., hard to approximate within a constant factor) even when demands are uniform \cite{garg1997primal}. \cite{chekuri2009unsplittable} obtained an $O(\log n)$-approximation. Recently, \cite{anagnostopoulos2014mazing} obtained a $2+\epsilon$-approximation for path topology.

In this work, we consider optimal power flow with inelastic demands as a mixed integer programming problem, which essentially generalizes CKP to a networked setting and UPF to consider complex-valued demands.

\section{Formulation and Notations}

We represent an electric distribution network by a graph $\cG=(\cV,\cE)$. The set of nodes $\cV$ denotes the electric buses, whereas the set of edges $\cE$ denotes the distribution lines. We index the nodes in $\cV$ by $\{0, 1..., |\cV|\}$, where the node $0$ denotes the generation source or the feeder to the main grid. 

A power flow in an alternating current (AC) electric network is characterized by a set of complex-valued quantities. Given a complex number $\psi \in \CC$, denote its real and imaginary components by $\psi^{\rm R} \triangleq \re(\psi)$ and $\psi^{\rm I} \triangleq \im(\psi)$ respectively, its complex conjugate by $\psi^\ast$, and its argument by ${\arg}(\psi)$.

Each node $i \in \cV\backslash\{0\}$ is associated with a load $\hat s_i \in \CC$. For node $i \in \cV$, we denote its voltage by $V_i \in \CC$. For each edge $e=(i,j) \in \cE$, we denote its current from $i$ to $j$ by $I_{i,j}$, its transmitted power by $S_{i,j}$, and its impedance by $z_{i,j} \in \CC$ (also denoted by $z_e$). A power flow in a steady state is described by a set of power flow equations: 
\begin{align} 
S_{i,j} & = V_i I_{i,j}^\ast,  \quad\forall (i,j) \in \cE \label{eq:1}\\
z_{i,j} I_{i,j} & =	V_i - V_j, \quad\forall (i,j) \in \cE  \label{eq:2}  \\
\hat s_j & = \sum_{i: (i,j) \in \cE}(S_{i,j} -  z_{i,j} |I_{i,j}|^2 ) - \sum_{l: (j,l)\in \cE} S_{j,l}, \ \forall j \in \cV  \label{eq:3} 
\end{align}

\subsection{Branch Flow Model for Radial Networks}

In particular, when $\cG$ is a radial network (i.e., a tree), node $0$ denotes the root of $\cG$. Without loss of generality, we assume that the root $0$ has only a single child, and the edge from the root to the child is denoted by $e_1$. 
If we denote an edge by a tuple $(i,j)$, then $i$ is referred to as the parent of $j$ (i.e., $i$ is the immediate upstream node from $j$ to the root $0$). Hence, Eqn.~\raf{eq:3} can be simplified as:
\begin{align}
& S_{i,j}  = \sum_{l: (j,l)\in \cE} S_{j,l} +\hat  s_j + z_{i,j} |I_{i,j}|^2  \quad\forall (i,j) \in \cE \\
& \hat s_0 + \sum_{j: (0,j)\in \cE}S_{0,j} = 0
\end{align}

As a result, the power flow equations (Eqns.~\raf{eq:1}-\raf{eq:3}) can be reformulated by the Branch Flow Model (BFM).
Let $v_i \triangleq |V_i|^2$ and $\ell_{i,j} \triangleq |I_{i,j}|^2$ be the magnitude square of voltage and current respectively. The BFM is given by the following:
\begin{align}
\ell_{i,j} &= \frac{|S_{i,j}|^2}{v_i}, \quad\forall (i,j) \in \cE \\
v_j &= v_i + |z_{i,j}|^2 \ell_{i,j} - 2 \re(z_{i,j}^\ast  S_{i,j}), \quad\forall (i,j) \in \cE \label{eq:vj}\\
S_{i,j} &= \sum_{l: (j,l)\in \cE} S_{j,l} + \hat s_j + z_{i,j}\ell_{i,j}, \quad\forall (i,j) \in \cE \label{eq:bf3}\\
\hat s_0 &=- \sum_{j: (0,j)\in \cE}S_{0,j}
\end{align} 
\ifsupplementary
   A formal proof can be found in the appendix.
\fi

\subsection{Utility Maximizing Optimal Power Flow Problem}
For each node $i\in\cV\backslash\{0\}$, there is a set of customers attached to $i$, denoted by $\cN_i$. Let $\cN \triangleq \cup_{i\in\cV\backslash\{0\}} \cN_i$ be the set of all customers. Among the customers, some have inelastic power demands, denoted by $\cI \subseteq\cN$, which are required to be either completely satisfied or curtailed. An example is an appliance that can be either switched on or off. The rest of customers, denoted by $\cF \triangleq \cN\backslash\cI$, have elastic demands, which can be partially satisfied. See an illustration in Fig.~\ref{fig:tree}.

Each customer $k\in \cN$ is associated with a tuple $(s_k, u_k)$, where $s_k\in \CC$ is a complex-valued demand, and $u_k\in \RR^+$ is the utility value when $k$'s demand ($s_k$) is completely satisfied. For customers with elastic demands, we assume that the utility value is proportional to the fraction of satisfied demand.   
We assign a control variable $x_k$ to each customer $k\in \cN$. If $k \in \cI$, then $x_k\in\{0,1\}$. Otherwise, if $k \in \cF$, then $x_k \in [0,1]$.

We observe that $\hat s_j = \sum_{k \in \cN_j} s_k x_k$. Hence, Eqn.~\raf{eq:bf3} can be reformulated as:  
\begin{equation}
	S_{i,j} = \sum_{l: (j,l)\in \cE} S_{j,l} + \sum_{k \in \cN_j}s_k x_k + z_{i,j}\ell_{i,j}
\end{equation}

Let $v_{\min}, v_{\max} \in\RR^+$ be the minimum and maximum allowable voltage magnitude square at any node, and $C_{i,j}\in \RR^+$ be the maximum allowable apparent power on edge $(i,j) \in \cE$. We consider two common operating constraints:
\begin{itemize}

\item ({\em Power Capacity Constraints}): $|S_e|\le C_e, \ \forall e \in \cE$.

\item ({\em Voltage Constraints}): $v_{\min} \le v_{j} \le v_{\max}, \ \forall j \in \cV$.

\end{itemize}

The goal of demand response is to decide a solution of control variables $(x_k)_{k \in \cN}$ that maximizes the total utility of satisfiable customers subject to the operating constraints. 
We define a utility maximizing optimal power flow (\textsc{MaxOPF}) by the following mixed integer programming problem.
\begin{align}
&\textsc{Input}:\; v_0;v_{\min}; v_{\max}; (u_k,s_k)_{k\in \cN}; (z_{i,j},C_{i,j})_{(i,j)\in \cE} \notag \\
&\textsc{Output}:\; s_0 ; (v_i)_{i\in \cV}; (S_{i,j}, \ell_{i,j})_{(i,j)\in \cE}; (x_k)_{k \in \cN} \notag \\
&\textsc{(MaxOPF)}\quad \max_{\substack{x_k, v_i,\ell_{i,j}, S_{i,j} \;\;}} \sum_{k \in \cN} u_k x_k,  \notag \\
\text{s.t.} \ & 	\ell_{i,j} = \frac{|S_{i,j}|^2}{v_i},  \qquad \forall (i,j) \in \cE  \label{eq:c0}\\
&  S_{i,j}=  \sum_{k \in \cN_j} s_k x_k  + \sum_{l: (j,l)\in \cE} S_{j,l} + z_{i,j}\ell_{i,j}, \ \forall (i,j) \in \cE  \label{eq:S}\\
& s_0 =-S_{0,1}\label{eq:S0}\\ 
&	v_j = v_i + |z_{i,j}|^2 \ell_{i,j} - 2 \re(z_{i,j}^\ast  S_{i,j}), \quad  \forall (i,j) \in \cE \label{eq:cv}\\
& |S_{i,j}| \le C_{i,j} ,  \qquad \forall (i,j) \in \cE \label{eq:c1}\\
& v_{\min} \le v_j \le v_{\max},  \quad  \forall j \in\cV\backslash \{0\}\label{eq:c2}\\
& x_k \in \{0,1\},\qquad \forall k \in \cI\\
& x_k \in [0,1], \qquad \forall k \in \cF\\
& v_i \in \RR^+, \ \ell_{i,j} \in \RR^+, \ S_{i,j} \in \CC \label{clast}
\end{align}

\begin{figure}[!htb]
\begin{center} 
	\includegraphics[scale=.6]{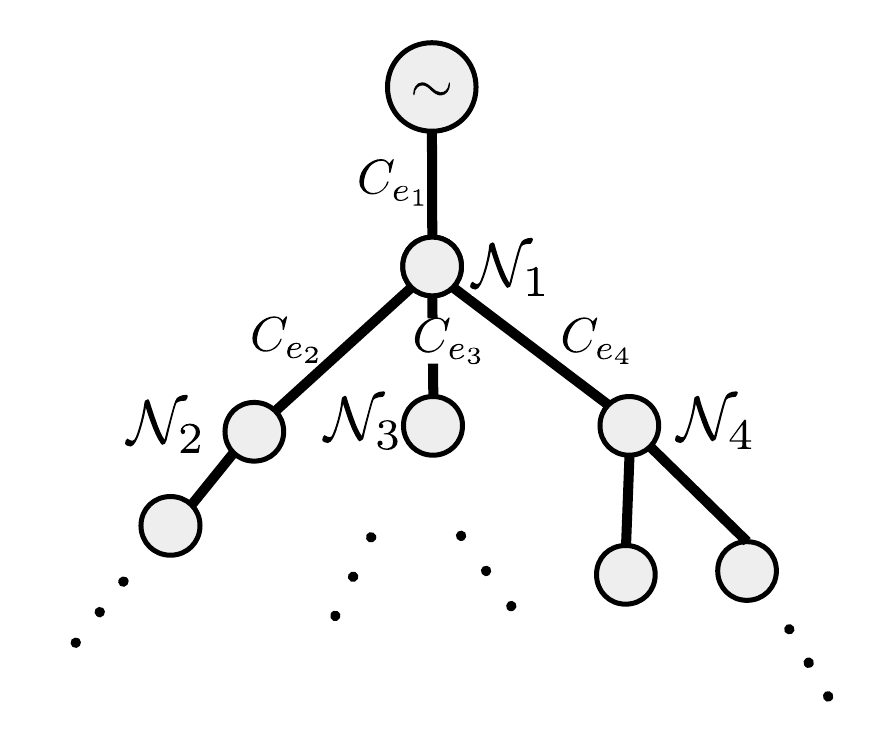}
\end{center}
\caption{An illustration of a radial network.}
\label{fig:tree}
\end{figure}  

In particular, we denote by {\sc MaxOPF$_{\sf C}$} when \textsc{MaxOPF} considers only power capacity constraints Cons.~\raf{eq:c1} without Cons.~\raf{eq:c2}, whereas by {\sc MaxOPF$_{\sf V}$} when \textsc{MaxOPF} considers only voltage constraints Cons.~\raf{eq:c2} without Cons.~\raf{eq:c1}. 

Note that a similar problem has been studied in \cite{li2012optimization}. But there are several differences: (1)  \cite{li2012optimization} considers only elastic demands, whereas we consider a mix of elastic and inelastic demands. (2)  \cite{li2012optimization} does not consider a power capacity constraint on each edge. (3)  \cite{li2012optimization} also considers transmission power loss in the objective function, whereas we consider a simpler problem that maximizes the total utility. Our problem provides the foundation for solving the more general problem. 

We observe that Cons.~\raf{eq:S0} is always satisfied and can be removed. As in \cite{gan2015exact}, we assume that the resistance and reactance for any edge is strictly positive (i.e.,  $z_e^{\rm R}, z_e^{\rm I} > 0$).

Note that \textsc{MaxOPF} is a difficult problem even when all demands are elastic (i.e., $\cN = \cF$), because of the non-convexity of Cons.~\raf{eq:c0}. Several convex relaxations have been shown to be exact under certain conditions \cite{gan2015exact,low2014convex1,low2014convex2}. The presence of inelastic demands makes the problem much harder (even for an approximate solution).

\subsection{Approximation Solutions}

Given a solution $x \triangleq (x_k)_{k \in \cN}$, we denote the total utility by $u(x)\triangleq \sum_{k \in \cN} u_k x_k$. We denote an optimal solution of \textsc{MaxOPF} by $x^\ast$ and $\OPT \triangleq u(x^\ast)$.

\begin{definition}
For $\alpha\in(0,1]$ and $\beta\ge1$, we define a bi-criteria $(\alpha,\beta)$-approximation to \textsc{MaxOPF} as a solution $\hat x =\big((\hat{x}_k)_{k \in \cI}, (\hat x_k)_{k\in\cF}\big) \in \{0, 1\}^{|\cI|}\times[0,1]^{|\cF|}$ satisfying
\begin{align}
& \text{\em Cons.~\raf{eq:c0}, \raf{eq:S}, \raf{eq:S0}, \raf{eq:cv}} \notag \\ 
& |S_{i,j}| \le \beta C_{i,j} ,  \qquad \forall (i,j) \in \cE \label{eq:betac1}\\
& \frac{1}{\beta} v_{\min} \le v_j \le \beta v_{\max},  \quad  \forall j \in\cV\backslash \{0\}\label{eq:betac2}
\end{align}
such that $u(\hat x) \ge \alpha \OPT$.
\end{definition}
For {\sc MaxOPF$_{\sf V}$}, Cons.~\raf{eq:betac1} will be dropped. For {\sc MaxOPF$_{\sf C}$}, Cons.~\raf{eq:betac2} will be dropped. 

In the above definition, $\alpha$ characterizes the approximation gap between an approximate solution and the optimal solution, whereas $\beta$ characterizes the violation bound of constraints. When $\beta=1$, an $(\alpha,\beta)$-approximation will be simply called an $\alpha$-approximation.


\section{Hardness Results}\label{sec:hard}

In this section, we present the hardness results of \textsc{MaxOPF}, which show that the general form of \textsc{MaxOPF} is hard to approximate. To study the hardness of approximation, we consider that all demands are inelastic (i.e., $\cN= \cI$ and $n = |\cI|$). First, we show that \textsc{MaxOPF} considering only voltage constraints is inapproximable by any efficient algorithm for any approximation gap and any violation bound polynomial in $n$, even for a one-edge network. Second, we show that \textsc{MaxOPF} with only power capacity constraints is inapproximable by any efficient algorithm in general (non-tree) networks, even in purely resistive electric networks. These hardness results motivate us to develop approximation algorithms for relaxed problem versions in the next section. 
\ifsupplementary
The proofs can be found in the appendix.
\fi

\subsection{Hardness of \textsc{MaxOPF}$_{\sf V}$}

\begin{theorem} \label{thm:OPFV-hard}
Unless {\rm P}={\rm NP}, there is no $(\alpha,\beta)$-approximation for {\sc MaxOPF$_{\sf V}$} (even when $|\cE| = 1$) by a polynomial-time algorithm in $n$, for any $\alpha$ and $\beta$ that have polynomial length in $n$.
\end{theorem}

\begin{remark} 
Theorem~\ref{thm:OPFV-hard} implies that when both upper and lower bounds of voltage ($v_{\min} \le v_j \le v_{\max}$) are considered, no practical approximation algorithm for {\sc MaxOPF$_{\sf V}$} exists. Hence, one has to relax the voltage constraints in order to obtain practical algorithms. 
In particular, Theorem~\ref{thm:OPFV-hard} holds whenever $\re(z^\ast_e s_k)$ is allowed to be arbitrary. If all demands are inductive (i.e., $\re(z^\ast_e s_k)\ge 0$), then the hardness result does not necessarily hold. 
{\color{black} Also, by a slight modification in the proof of the theorem, Theorem~\ref{thm:OPFV-hard} holds when more than one generator (or capacitor) is allowed, because  a generator (or capacitor) is associated with negative power in our formulation.} 
\end{remark}

\subsection{Hardness of \textsc{MaxOPF}$_{\sf C}$}

Next, we consider \textsc{MaxOPF}$_{\sf C}$ in general (non-tree) networks. The formulation of \textsc{MaxOPF} will be slightly modified to accommodate a general topology. More precisely, Cons.~\raf{eq:S}-\raf{eq:S0} are replaced by:
\begin{equation}
\sum_{i: (i,j) \in \cE}(S_{i,j} -  z_{i,j} |I_{i,j}|^2 ) - \sum_{k: (j,k)\in \cE} S_{j,k} = \sum_{k\in \cN _j}  s_k x_k,  \quad\forall j \in \cV \label{eq:genS}
\end{equation}
and Cons.~\raf{eq:c0}, \raf{eq:cv} by \raf{eq:1}, \raf{eq:2}. 

\medskip
\begin{definition}
For $\alpha\in(0,1]$ and $\beta\ge 1$, we define a bi-criteria $(\alpha,\beta)$-approximation to \textsc{MaxOPF}$_{\sf C}$ as a solution $\hat x =\big((\hat{x}_k)_{k \in \cI}, (\hat x_k)_{k\in\cF}\big) \in \{0, 1\}^{|\cI|}\times[0,1]^{|\cF|}$ satisfying
\begin{align}
& \text{\em Cons.~\raf{eq:1}, \raf{eq:2}, \raf{eq:genS}} \notag \\ 
& |S_{i,j}| \le \beta C_{i,j} ,  \qquad \forall (i,j) \in \cE 
\end{align}
such that $u(\hat x) \ge \alpha \OPT$.
\end{definition}

\begin{theorem} \label{thm:OPFC-hard}
Unless P=NP, there exists no ($\alpha, \beta$)-approximation for \textsc{MaxOPF}$_{\sf C}$ in general networks, for any $\alpha$ and $\beta$ having  polynomial length in $n$, even in purely resistive electric networks (i.e., ${\rm Im}(z_{i,j})= 0$ for all $(i, j) \in \cE$ and ${\rm Im}(s_k) = 0$ for all $k \in {\cal N}$).
\end{theorem}

\begin{remark} The ramification of Theorem~\ref{thm:OPFC-hard} is that \textsc{MaxOPF}$_{\sf C}$ is inapproximable in general networks (e.g., topologies with cycles). To derive practical algorithms, one has to consider acyclic topologies (i.e., trees).
 
There are other papers which studied NP-hardness of AC power flow problems \cite{LGH16AChard,verma09thesis,bv15AChard}. But there are several differences compared to our hardness results. The results in \cite{LGH16AChard,verma09thesis,bv15AChard} consider a different set of constraints, namely the phase angle difference on each link is bounded by some threshold, in addition to the voltage constraints. In our paper, we consider, either voltage constraints alone, or power capacity constraints alone (the latter can be related to phase angle constraints). The setting in \cite{LGH16AChard,verma09thesis,bv15AChard} with no binary variables implies that checking feasibility is already NP-hard; on the other hand, since we allow binary variables associated with loads in our setting, the all-zero solution (and all voltages equal to $v_{\text min}$ in case $v_{\text min}>0$) is trivially feasible. In this case, the non-trivial question is about optimization rather than decision. {\color{black}  While the results in \cite{LGH16AChard,verma09thesis,bv15AChard} show NP-hardness of (continuous) AC feasibility, we study the discrete problem. We show hardness of approximation, even if we allow the capacity/voltage constraints to be violated by some multiplicative parameter $\beta$. }
\end{remark}
\section{Approximation Algorithms}\label{sec:ind}

In this section, we present an approximation algorithm for \textsc{MaxOPF}. Motivated by the hardness results in the last section, we relax \textsc{MaxOPF} to consider one-sided voltage constraints and tree topology. First, we define a simplified model (called \textsc{sMaxOPF}) assuming bounded transmission power loss. We then provide an efficient approximation algorithm for \textsc{sMaxOPF} considering all inelastic demands, based on an approximation algorithm for the unsplittable flow problem. We then analyze the approximation ratio of our approximation algorithm. Next, we adapt our algorithm to \textsc{MaxOPF} considering a mix of inelastic and elastic demands.

\subsection{\mbox{Simplified Utility Maximizing Optimal Power Flow Problem}}

In this section, we consider a simplified model in tree topology (which is related to linear DistFlow model proposed in \cite{baran1989placement}, \cite{baran1989sizing}). The basic idea is that the transmission power loss in an electric network is usually small. One can approximate the optimal power flow model by upper bounding the terms associated with transmission power loss (i.e., $|z_{e}|^2 \ell_{e}, z_{e}\ell_{e}$), and then derive a feasible solution without explicitly considering the transmission power loss. 

First, denote by $P_e \subseteq \cE$ the path from edge $e$ to the root $0$, and by $P_k \subseteq \cE$ the path from node $k$ to the root $0$.
Note that $\ell_{i,j} = \frac{|S_{i,j}|^2}{v_i}$ is positive. We assume that $\ell_e \le \bar\ell_e$ for a constant $\bar\ell_e$ independent of solution $(x_k)_{k \in {\cal N}}$.

We rewrite Eqn.~\raf{eq:S} in \textsc{MaxOPF} by recursively substituting $S_{j,l}$:
\begin{equation}
S_e = \sum_{k: e  \in P_k} s_k x_k + \sum_{e': e \in P_{e'}} z_{e'}  \ell_{e'} \label{eq:se}
\end{equation}
Note that 
\begin{align*}
 |S_e| & \le
\bigg|\sum_{k: e  \in P_k} s_k x_k\bigg| + \bigg|\sum_{e': e \in P_{e'}} z_{e'}  \ell_{e'}\bigg| \\
& \le \bigg|\sum_{k: e  \in P_k} s_k x_k\bigg| + \bigg|\sum_{e': e \in P_{e'}} z_{e'} \bar\ell_e\bigg| 
\end{align*}
Let $\hat{L}_e \triangleq \big|\sum_{e': e \in P_{e'}} z_{e'}  \bar\ell_e\big|$. Thus, the power capacity constraint Cons.~\raf{eq:c1} can be implied by the following constraint:
$$
\bigg| \sum_{k: e  \in P_k} s_k x_k \bigg| \le \hat{C}_e \triangleq \max\{ C_e - \hat{L}_e, 0 \}
$$

Also, we rewrite Cons.~\raf{eq:cv} recursively by substituting $v_i$:
$$
v_j = v_0 -2\sum_{ e' \in P_{e} } \re(z_{e'}^\ast  S_{e'}) +  \sum_{ e' \in P_{e} } |z_{e'}|^2 \ell_{e'},
$$
where $e=(i,j)$. Hence, Cons.~\raf{eq:c2} becomes:
\begin{align}
&\tfrac{1}{2}(v_0 - v_{\max}  + \sum_{ e' \in P_{e} } |z_{e'}|^2 \ell_{e'} \notag) \\
\le \ &   \sum_{ e' \in P_{e} } \re(z_{e'}^\ast  S_{e'})  \le \tfrac{1}{2}(v_0 - v_{\min}+ \sum_{ e' \in P_{e} } |z_{e'}|^2 \ell_{e'} ) \label{eq:c2-e}
\end{align}
Let $\underline{V_e} \triangleq \tfrac{1}{2}(v_0 \mbox{-} v_{\max} + \sum_{ e' \in P_{e} } |z_{e'}|^2 \bar\ell_e)$ and $\overline{V_e} \triangleq \tfrac{1}{2}(v_0 \mbox{-} v_{\min})$.
Thus, voltage constraint Cons.~\raf{eq:c2-e} can be implied by the following constraints:
\begin{align*}
\underline{V_e} \le & \sum_{ e' \in P_{e} } \sum_{k: e'  \in P_k} (z^{\rm R}_{e'} s_k^{\rm R}+ z^{\rm I}_{e'} s_k^{\rm I} ) x_k \le \overline{V_e}, \\
\Leftrightarrow \quad \underline{V_e} \le & \sum_{k \in \cN} \Big( \sum_{ e'  \in P_k \cap P_e} z^{\rm R}_{e'} s_k^{\rm R}+ z^{\rm I}_{e'} s_k^{\rm I} \Big) x_k \le \overline{V_e},
\end{align*}
where the second statement follows from exchanging the summation operators.

Therefore, we define a simplified utility maximizing optimal power flow problem (\textsc{sMaxOPF}), such that a feasible solution to \textsc{sMaxOPF} is a feasible solution to \textsc{MaxOPF}.
\begin{align}
&\textsc{(sMaxOPF)}\quad \max_{x_k} \sum_{k \in \cN} u_k x_k \notag\\
\text{s.t.} \ 	& \bigg| \sum_{k: e  \in P_k} s_k x_k \bigg| \le \hat{C}_{e}, \qquad \forall  e \in \cE\label{con} \\
&\underline{V_e} \le \sum_{k \in \cN} \Big( \sum_{ e'  \in P_k \cap P_e} z^{\rm R}_{e'} s_k^{\rm R}+ z^{\rm I}_{e'} s_k^{\rm I} \Big) x_k \le \overline{V_e}, \ \forall  e \in \cE \label{conV}\\
& x_k \in \{0,1\},\quad \forall k \in \cI\\
& x_k \in [0,1], \quad \forall k \in \cF 
\end{align}

If we assume $z_{e} \ell_{e} \to 0$ and $|z_{e}|^2 \ell_{e} \to 0$, this is known as the linear DistFlow model proposed in \cite{baran1989placement}, \cite{baran1989sizing}.

We denote by {\sc sMaxOPF$_{\sf C}$} when \textsc{sMaxOPF} considers only power capacity constraints Cons.~\raf{con} without Cons.~\raf{conV}, whereas by {\sc sMaxOPF$_{\sf V}$} when \textsc{sMaxOPF} considers only voltage constraints Cons.~\raf{conV} without Cons.~\raf{con}. 

\subsection{Approximation Algorithm for \textsc{sMaxOPF} with All Inelastic Demands}

\begin{algorithm}[ht!]
\caption{\mbox{\sc InelasDemAlloc}$[(u_k,s_k)_{k \in \cN}]$}
\label{alg1}
\begin{algorithmic}[1]
	\Require customers' utilities $(u_k)$ and inelastic demands $(s_k)$
	\State Let $L \triangleq \frac{u_{\max}}{n^2}$ and $u_{\max} = \max_{k\in\cN}u_k$
	\State Let $\bar u_k \triangleq \big\lfloor\frac{u_k}{L}\big\rfloor$ for each $k \in \cN$
	\Statex \Comment{Group customers according to the range of their utilities}
	\State $\hat\cN_1 \leftarrow \{k \in \cN \mid \bar u_k \in [0, 2)\}$ \label{alg:partition0}
	\For {$i = 2,...,\lceil 2\log n \rceil + 1 $}
	\State $\hat\cN_i \leftarrow \{k \in \cN \mid \bar u_k \in [2^{i-1}, 2^{i})\}$ \label{alg:partition}		
	\EndFor		
	\Statex \Comment{Call {\sc GreedyAlloc} to solve sub-problems with $\hat\cN_i$}
	\For {$i = 1,...,\lceil 2\log n \rceil+1 $}
	\State $M_i \leftarrow  \mbox{\sc GreedyAlloc} [(s_k)_{k \in \hat\cN_i}]$ \label{alg:up}
	\EndFor
	\Statex \Comment{Return the group with maximum utility}
	\State \Return $M$ such that $u(M)=\displaystyle \max_{i = 1,...,\lceil 2\log n \rceil +1} u(M_i)$
\end{algorithmic}
\end{algorithm}

\begin{algorithm}[htb!]
\caption{$\mbox{\sc GreedyAlloc} [(s_k)_{k \in \hat\cN}]$}\label{alg2}
\begin{algorithmic}[1]
\Require customers' inelastic demands $(s_k)$

\State Sort customers in $\hat\cN$ according to the magnitudes of demands: 
\Statex \qquad\qquad  $|s_1| \le |s_2| \le ... \le \big|s_{|\hat\cN|}\big|$

\State $M \leftarrow \varnothing$
\For{each $k \in \hat\cN$}
\State {\bf if} $\displaystyle \bigg|\sum_{k' \in M \cup \{k\}: e \in P_{k'}} s_{k'} \bigg| \le \hat{C}_e, \ \forall e \in \cE$ {\bf and}  \label{alg:gd_feasible2} 
\Statex \  \ 
$\displaystyle \underline{V_e} \le \sum_{ e' \in P_{e} } \sum_{k' \in M \cup \{k\}: e'\in P_{k'}} z^{\rm R}_{e'} s_{k'}^{\rm R}+ z^{\rm I}_{e'} s_{k'}^{\rm I}\le \overline{V_e}, \ \forall e \in \cE$ 
\State {\bf then} $ M \leftarrow M \cup \{k\}$ 
\EndFor
			
\State \Return $M$
\end{algorithmic}
\end{algorithm}

In this section, we provide an approximation algorithm  ({\sc InelasDemAlloc}) to \textsc{sMaxOPF} considering all inelastic demands (i.e., $\cN= \cI$). This algorithm is inspired by an $O(\log n)$-approximation algorithm for the unsplittable flow problem in \cite{CKS06}. 

Algorithm~\ref{alg1} ({\sc InelasDemAlloc}) first normalizes the customers' utilities by $\bar u_k \triangleq \big\lfloor \frac{u_k}{L} \big\rfloor$. Then it partitions customers into groups $(\hat\cN_1,...,\hat\cN_{\lceil 2\log n \rceil  + 1})$ according to the ranges of normalized utilities, such that the utilities of the $i$-th group are within $[2^{i-1}, 2^{i})$. For each group, it next calls {\sc GreedyAlloc} to return a feasible solution for the group of customers. Finally, {\sc InelasDemAlloc} returns the output solution as the group from {\sc GreedyAlloc} with the maximum utility. 

Algorithm~\ref{alg2} ({\sc GreedyAlloc}) first sorts the customers in a non-decreasing order according to the magnitudes of their demands. Then, it packs their demands greedily sequentially in that order, if the power capacity constraints or voltage constraints are not violated. The customers who can be satisfied are placed in the set $M$.

Evidently, both {\sc InelasDemAlloc} and {\sc GreedyAlloc} have polynomial running time in $n$. 

\subsection{Analysis of Approximation Ratio for \textsc{sMaxOPF}}
 
We first provide an intuition of {\sc InelasDemAlloc} and {\sc GreedyAlloc}. {\sc InelasDemAlloc} groups the customers with similar utilities, whereas {\sc GreedyAlloc} finds a solution that maximizes the number of satisfied customers greedily. If {\sc GreedyAlloc} can find a solution that is close to the optimal solution, when all customers have the same utility, then {\sc InelasDemAlloc} can find a group that approximates the optimal solution in general.

We denote by \textsc{GreedyAlloc$_{\sf C}$} when solving {\sc sMaxOPF$_{\sf C}$} (i.e., $\underline{V_e} \to - \infty$ and $\overline{V_e} \to \infty$), and by \textsc{GreedyAlloc$_{\sf V}$} when solving {\sc sMaxOPF$_{\sf V}$} (i.e., $\hat{C}_e \to \infty$ for all $e \in \cE$).

\ifsupplementary
In the appendix,
\fi 
we also show that {\sc sMaxOPF$_{\sf V}$} with both upper and lower voltage constraints are also inapproximable by any efficient algorithm for any approximation gap and any violation bound polynomial in $n$. Hence, we drop the lower voltage constraints ($\underline{V_e}$) as we analyze the approximation ratio of {\sc GreedyAlloc}.

\medskip

\noindent
{\em C.1. Analysis of {\sc GreedyAlloc}}

Although {\sc GreedyAlloc} resembles an $O(\log n)$-approximation algorithm for unsplittable flow problem provided in \cite{CKS06}, our proof for the approximation ratio is substantially more involved than that in \cite{CKS06}, because of the presence of complex-valued demands makes {\sc GreedyAlloc} behave very differently. 

To analyze the approximation ratio of {\sc GreedyAlloc}, we first consider a simple setting where all utilities are identical (i.e., $u_k = 1$ for all $k\in \cN$). The objective of \textsc{sMaxOPF} then becomes to maximize the number of satisfied customers.


\medskip

We will define the following notations:
\begin{itemize}

\item
A demand path is a path from a customer to the root.
Let $\eta \triangleq \max_{e \in \cE} |P_e|$ be the maximum length of any demand path.

\item
Let $\theta \triangleq \max_{k,k' \in \cN} |\arg(s_k) - \arg(s_{k'})|$ be the maximum angle difference between any pair of demands

\item
Let $\theta_{\rm zs} \triangleq \max_{k \in \cN, e \in P_k} |\arg(s_k) - \arg(z_e)|$ be the maximum angle difference between demands and line impedance along any path to the root. We assume $0\le \theta_{\rm zs} < \tfrac{\pi}{2}$.

\item
Let $\rho_{{\tilde{e}}} \triangleq \max_{e,e' \in P_{{\tilde{e}}}} \tfrac{|z_e|}{|z_{e'}|}$ be the maximum ratio of impedance magnitude between any pair of edges along the path $P_{{\tilde{e}}}$.

\item
Let $\rho \triangleq \max_{{\tilde{e}} \in \cE} \rho_{{\tilde{e}}}$ be the maximum of all ratios.

\end{itemize}
Since $0 \le \theta_{\rm zs} <\tfrac{\pi}{2}$, it necessarily holds that $z^{\rm R}_{e} s_k^{\rm R}+ z^{\rm I}_{e} s_k^{\rm I}\ge 0$, for all $k\in \cN$ and $e\in \cE$. 
It follows that  Cons.~\raf{conV} on edge $e$ is at least as large as when $e \in \cL$ is a leaf edge, where $\cL$ is the set of all leaf edges defined by: 
$$
\cL\triangleq \{(i,j) \in \cE \mid \nexists k \in \cV \text{ such that } (j,k) \in \cE\}
$$

Therefore, it suffices to consider Cons.~\raf{conV} for each $e\in \cL$:
\begin{equation}
\sum_{k \in \cN} \Big( \sum_{ e'  \in P_k \cap P_e} z^{\rm R}_{e'} s_k^{\rm R}+ z^{\rm I}_{e'} s_k^{\rm I} \Big) x_k\le \overline{V_e}, \ \ \forall e \in \cL \label{conV2}
\end{equation}
\medskip

\begin{theorem} \label{thm:greedyalloc}
Consider $u_k=1$ for all $k\in \cN$ and assume $0\le \theta,\theta_{\rm zs} < \tfrac{\pi}{2}$.
\begin{enumerate}
\item  {\sc GreedyAlloc$_{\sf C}$} is $\alpha$-approximation for {\sc sMaxOPF$_{\sf C}$}, where 
$$\alpha =  \Big(  \big\lfloor   \sec \theta \cdot \sec\tfrac{\theta}{2}    \big\rfloor + 1 \Big)^{-1}$$

\item {\sc GreedyAlloc$_{\sf V}$} is $\alpha$-approximation for {\sc sMaxOPF$_{\sf V}$}, where
$$\alpha =  \Big(\big\lfloor   \eta\cdot  \rho \cdot\sec\theta_{\rm zs} \big\rfloor + 1  \Big)^{-1}$$ 	

\item  {\sc GreedyAlloc} is $\alpha$-approximation for \textsc{sMaxOPF}, where
$$\alpha =  \Big(\big\lfloor   \eta\cdot  \rho \cdot\sec\theta_{\rm zs} \big\rfloor + \big\lfloor   \sec \theta \cdot \sec\tfrac{\theta}{2}\big\rfloor + 2 \Big)^{-1}$$ 
\end{enumerate}
\label{thm:1}
\end{theorem}

\begin{proof}
We first present the basic idea as follows. {\sc GreedyAlloc} first sorts customers in $\hat{\cN}$ in a non-decreasing order according to the magnitudes of their demands:
$$
\overbrace{|s_1| \le |s_2| \le ...}^{A_1} \le \overbrace{|s_t| \le ...}^{B_1} \le \overbrace{|s_{p}| \le ...}^{A_2} \le ...
\le \overbrace{|s_{z}| \le ...}^{B_m} 
$$
We index the customers in the solution set $M $ as $ \{k_1, k_2,...,k_r\}$. {\sc GreedyAlloc} attempts to pack their demands greedily sequentially (by placing the satisfied customers into $M$), if the power capacity constraints or voltage constraints are not violated. Let the sets of customers who can be satisfied consecutively be $(A_1,...,A_m)$ and the sets of customers who violate Con.~\raf{con} or \raf{conV2} be $(B_1,...,B_m)$ (where $B_m$ may be empty).

Let the optimal solution be $R^\ast \subseteq \hat{\cN}$ as the maximal set of satisfied customers. We follow an exchange argument by constructing two sequences of sets $(W_1,...,W_r)$ and $(R_1,...,R_r)$, such that the following conditions hold: 
\begin{enumerate}

\item $W_j \subseteq R^\ast \cap  \bigcup_{i=1}^m B_i$ for each $j = 1,...,r$.

\item Set $R_j \triangleq(R_{j-1} \backslash  W_{j}) \cup  \{k_j\}$ for each $j = 1,...,r$.

\item Finally, we obtain $R_r = M$.

\end{enumerate}
The size of each $|W_j|$ will be used to derive the approximation ratio $\alpha$.

Formally, we define 
$$
Q_k(e) \triangleq \sum_{e' \in P_k \cap P_{ e}}  (z_{e'}^{\rm R} s_{k}^{\rm R} + z_{e'}^{\rm I} s_{k}^{\rm I} )
$$
and Cons.~\raf{conV2} becomes
$$
 \sum_{k \in \cN} Q_k(e) x_k \le \overline{V_e}, \ \ \forall e \in \cL 
$$
Let $R_1 \triangleq (R^\ast \backslash W_1) \cup \{k_1\}$, such that customer $k_{1}\in M$ is added to $R^\ast$, and $W_1$ is removed. Recursively, define $R_j \triangleq(R_{j-1} \backslash  W_{j}) \cup  \{k_j\} $ for $j=2,...,r$.  For each step $j$, $W_j\subseteq R_{j-1}\cap  \bigcup_{i=1}^m B_i$ is defined to be any {\em minimal}  subset such that $R_j$ is a feasible solution. 

Define $W_j^{(1)} \subseteq W_j$ and $W_j^{(2)} \subseteq W_j$ as follows.
\begin{align} \hspace{-10pt}
W_j^{(1)} \triangleq \Bigg\{  k \in W_j \ \ \bigg| & \ \ \exists e \in \cE \notag, \ \ \bigg| \sum_{k' \in R_{j} \wedge e \in P_{k'}} s_{k'}  \bigg| \le \hat{C}_e \\
&  \  \mbox{and\ } 
\bigg| \sum_{k' \in R_{j} \cup\{k\} \wedge e \in P_{k'}} s_{k'} \bigg| > \hat{C}_e
\Bigg\} \label{eq:v1}
\end{align}
\begin{align} 
W_j^{(2)} \triangleq \Bigg\{  k \in W_j \ \ \bigg| & \ \ \exists e \in \cL \notag, \ \ \sum_{k' \in R_j} Q_{k'}(e) \le \overline{V_e} \\
&  \qquad \mbox{and\ } 
\sum_{k' \in R_j\cup\{k\}} Q_{k'}(e)>\overline{V_e}
\Bigg\} \label{eq:v2}
\end{align}
See Fig.~\ref{fig:thm} for an illustration.

\begin{figure}[!htb]
	\begin{center}
		\includegraphics[scale=0.8]{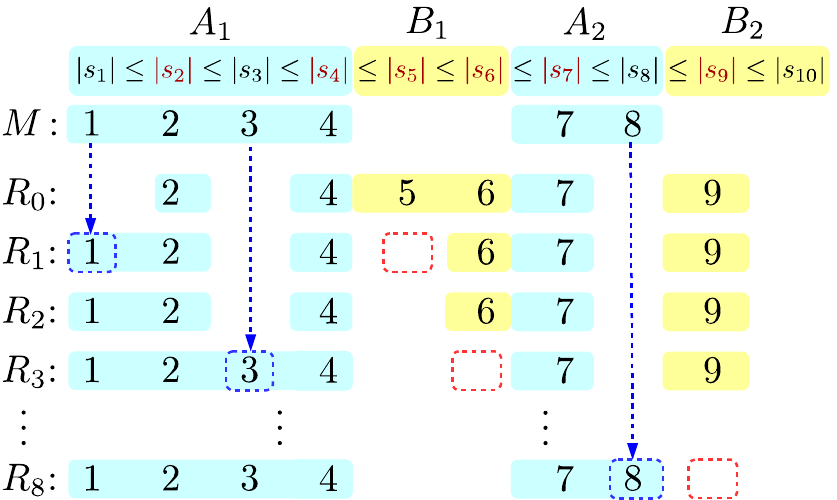}
	\end{center}
	\caption{A pictorial illustration of the definition of $R_j$ for an instance of $10$ customers $\hat\cN = \{1,2,...,10\}$. }
	\label{fig:thm}
\end{figure}
The approximation ratio $\alpha$ is equivalent to the following bound:
$$
|M| \ge \alpha |R^\ast|
$$
The proof is completed by Lemmas~\ref{lem:greedyalloc3}, \ref{lem:greedyalloc1}, \ref{lem:greedyalloc2}.
\end{proof}

\begin{lemma} \label{lem:greedyalloc3} 
\

\begin{enumerate}
\item  Consider {\sc GreedyAlloc$_{\sf C}$} for {\sc sMaxOPF$_{\sf C}$}. 
$$|M| \ge \frac{ |R^\ast| }{ \lfloor \sec \theta \cdot \sec\tfrac{\theta}{2}  \rfloor + 1}$$

\item Consider {\sc GreedyAlloc$_{\sf V}$} for {\sc sMaxOPF$_{\sf V}$}.
$$|M| \ge \frac{ |R^\ast| }{\lfloor  \eta\cdot  \rho \cdot\sec\theta_{\rm zs}  \rfloor + 1} $$	

\item Consider {\sc GreedyAlloc} for \textsc{sMaxOPF}.
$$|M| \ge \frac{ |R^\ast| }{\lfloor  \eta\cdot  \rho \cdot\sec\theta_{\rm zs}  \rfloor +  \lfloor \sec \theta \cdot \sec\tfrac{\theta}{2}  \rfloor + 2} $$
\end{enumerate} 
\end{lemma}
\begin{proof}
	By the definition of $W_j$, we observe the following:
	\begin{itemize}
		\item[({\sf o1})] $|W_j| \le |W_j^{(1)} | + |W_j^{(2)} |$.
		\item[({\sf o2})] If  $k_j \in R_{j-1}$, then  $|W_j| = 0$.
		\item[({\sf o3})] By ({\sf o2}), $(W_1,...,W_r)$ form a partition over $R^\ast \backslash M$.
	\end{itemize}
	Consider {\sc GreedyAlloc} for \textsc{sMaxOPF}.
	By Lemmas~\ref{lem:greedyalloc1} and \ref{lem:greedyalloc2} and ({\sf o3}), one can relate $|M|$ to the optimal $|R^\ast|$ as follows:
	\begin{align*}
	\ & |M| =  |M \cap R^\ast| + |M \backslash R^\ast| \\
	= \ &|M \cap R^\ast| +  \sum_{j \in M \backslash R^\ast} \frac{1}{|W_j|} |W_j| \\
	\ge  \ & |M \cap R^\ast| + \frac{ \sum_{j \in  M \backslash R^\ast} |W_j| }{\lfloor  \eta\cdot  \rho \cdot\sec\theta_{\rm zs}  \rfloor +  \lfloor \sec \theta \cdot \sec\tfrac{\theta}{2}  \rfloor + 2} \\
	= \ & |M \cap R^\ast| + \frac{ |R^\ast\backslash M| }{\lfloor  \eta\cdot  \rho \cdot\sec\theta_{\rm zs}  \rfloor +  \lfloor \sec \theta \cdot \sec\tfrac{\theta}{2}  \rfloor + 2}  \ \ (\mbox{by ({\sf o3})}) \\
	\ge \ & \frac{ |M \cap R^\ast| + |R^\ast\backslash M| }{\lfloor  \eta\cdot  \rho \cdot\sec\theta_{\rm zs}  \rfloor +  \lfloor \sec \theta \cdot \sec\tfrac{\theta}{2}  \rfloor + 2} \\
	= \ & \frac{ |R^\ast| }{\lfloor  \eta\cdot  \rho \cdot\sec\theta_{\rm zs}  \rfloor +  \lfloor \sec \theta \cdot \sec\tfrac{\theta}{2}  \rfloor + 2} \\
	\end{align*}
	The cases of {\sc GreedyAlloc$_{\sf C}$} and {\sc GreedyAlloc$_{\sf V}$} are special cases of $W_j = W_j^{(1)}$ and  $W_j = W_j^{(2)}$ respectively.
\end{proof}

\begin{lemma} \label{lem:greedyalloc1}
Define $W_j^{(1)}$ as in Eqn.~\raf{eq:v1}. We obtain:
\begin{equation}
|W_j^{(1)}| \le \lfloor \sec \theta \cdot \sec\tfrac{\theta}{2}  \rfloor+ 1 \label{eq:b3}
\end{equation}
\end{lemma}

\begin{lemma} \label{lem:greedyalloc2}
Define $W_j^{(2)}$ as in Eqn.~\raf{eq:v2}. We obtain:
\begin{equation}
|W_j^{(2)}| \le \lfloor \eta\cdot {\rho\cdot\sec\theta_{\rm zs}}\rfloor + 1 \label{eq:vpf3}.
\end{equation}
\end{lemma}

\bigskip

\noindent
{\em C.2. Analysis of {\sc InelasDemAlloc}}

We complete the analysis of {\sc InelasDemAlloc} by the following theorem.
\medskip

\begin{theorem}	\label{thm:chek}
Assume that $\rho$, $\sec\theta$ and $\sec\theta_{zs}$ are constants, and $\theta, \theta_{zs} < \tfrac{\pi}{2}$, then		
\begin{enumerate}
\item {\sc InelasDemAlloc$_{\sf C}$} is  $\frac{1}{O( \log n)}$-approximation for {\sc sMaxOPF$_{\sf C}$}.
\item {\sc InelasDemAlloc$_{\sf V}$} is $\frac{1}{O(\eta\cdot \log n)}$-approximation for {\sc sMaxOPF$_{\sf V}$}.
\item {\sc InelasDemAlloc} is $\frac{1}{O(\eta\cdot \log n)}$-approximation for \textsc{sMaxOPF}.
\end{enumerate}
\end{theorem}
\ifsupplementary
The proof can be found in the appendix.
\fi

\begin{remark}\color{black}
Basically, the approximation ratio is inversely proportional to the number of inelastic customers logarithmically, and the depth of  the electric network. The running time of {\sc InelasDemAlloc} is $O(n \log n + n \cdot \eta)$.
\end{remark}

\bigskip

\subsection{Approximation Algorithm for \textsc{MaxOPF} with Elastic and Inelastic Demands}

To incorporate elastic demands, we first solve the relaxed problem \textsc{rMaxOPF} by relaxing all inelastic demands to be elastic as follows:
\begin{align}
&\textsc{(rMaxOPF)}\quad \max_{\substack{x_k, v_i,\ell_{i,j}, S_{i,j} \;\;}} \sum_{k \in \cN} u_k x_k,  \notag \\
\text{s.t.} \ \ & 	\text{Cons.~\raf{eq:c0}, \raf{eq:S}, \raf{eq:S0}, \raf{eq:cv}, \raf{eq:c1}, \raf{eq:c2}} \notag\\
& x_k \in [0,1], \qquad \forall k \in \cN \notag
\end{align}

Note that Cons.~\raf{eq:c0} is non-convex and the problem is generally difficult to solve. Instead, we can consider a convex relaxation by relaxing the constraint to be $\ell_{i,j}\ge\frac{|S_{i,j}|^2}{v_i}$ as in  \cite{gan2015exact}. 
Let the solution be $\tilde{x} = \Big( (\tilde{x}_k)_{k \in \cF}, (\tilde{x}_k)_{k \in \cI} \Big)$. 


We next define a simplified residual problem \textsc{siMaxOPF$_\delta[\tilde x]$} by assuming the elastic demands are set according to $(\tilde{x}_k)_{k \in \cF}$, and the links capacity are reduced by a factor of $(1-\delta)$ for a given $\delta \in (0,1)$:
\begin{align}
&\textsc{(siMaxOPF$_\delta[\tilde x]$)}\quad \max_{x_k} \sum_{k \in \cN} u_k x_k \notag\\
\text{s.t.} \ & \bigg| \sum_{k: e  \in P_k} s_k x_k \bigg| \le (1-\delta)\cdot C_{e}, \qquad \forall  e \in \cE \\
&  \sum_{k \in \cN} \Big( \sum_{ e'  \in P_k \cap P_e} z^{\rm R}_{e'} s_k^{\rm R}+ z^{\rm I}_{e'} s_k^{\rm I} \Big) x_k \le \overline{V_e}, \ \forall  e \in \cE \\
& x_k \in \{0,1\},\quad \forall k \in \cI\\
& x_k = \tilde{x}_k, \qquad \forall k \in \cF
\end{align}

Then, we solve \textsc{siMaxOPF$_\delta[\tilde x]$} by {\sc InelasDemAlloc}.
To verify the feasibility of a solution $\bar x$ by {\sc InelasDemAlloc}, we consider the following problem with given demands $\bar x$:
\begin{align}
&\textsc{(OPF$[\bar x]$)}\quad \min_{\substack{v_i,\ell_{i,j}, S_{i,j} \;\;}} \sum_{e \in \cE} |z_e|\cdot \ell_e,  \notag \\
\text{s.t.} \ \ & 	\text{Cons.~\raf{eq:S}, \raf{eq:S0}, \raf{eq:cv}, \raf{eq:c1}, \raf{eq:c2}} \notag\\
&\ell_{i,j}\ge\frac{|S_{i,j}|^2}{v_i}\qquad \forall (i,j)\in \cE\notag\\
& x_k = \bar x_k , \qquad \forall k \in \cN \notag
\end{align}

\begin{algorithm}[ht!]
	\caption{\mbox{\sc MixDemAlloc}$[(u_k,s_k)_{k \in \cN}, \epsilon]$}
	\label{alg:lin}
	\begin{algorithmic}[1]
		\Require customers' utilities $(u_k)$ and inelastic demands; small step size $\epsilon\in (0,1)$
		\State $\delta \leftarrow 0$
		\Repeat
		\State $\tilde x\;\, \leftarrow$ Solution of {\sc rMaxOPF} \label{alg:m1}
		\State $M \leftarrow \textsc{InelasDemAlloc}$ on {\sc siMaxOPF}$_\delta[\tilde x]$ \label{alg:m2}
		\For{$k \in \cN$}
		\State $\bar x_k \triangleq \left\{\begin{array}{cc}
		1 & \text{if } k \in M \\
		\tilde x_k & \text{if } k \in \cF\\
		0 & \text{otherwise}
		\end{array}\right.$ 
		\EndFor
		\State $\delta \leftarrow \delta + \epsilon$
		\Until { {\sc OPF$[\bar x]$} is feasible }
		\State \Return $\bar x$		
	\end{algorithmic}
\end{algorithm}

We provide Algorithm~\ref{alg:lin} ({\sc MixDemAlloc}) as an efficient method to obtain a feasible solution to {\sc maxOPF} with both inelastic and elastic demands.
\begin{remark}\color{black}
	The running time of {\sc MixDemAlloc} is $O(\tfrac{1}{\delta}(n \log n + n \cdot \eta + T))$ where $T$ is the running time of solving {\sc OPF$[\bar x]$} via convex optimization  \cite{farivar2013branch}.
\end{remark}
  The theoretical approximation ratio of {\sc MixDemAlloc} with respect to {\sc maxOPF} is hard to obtain. But the empirical performance of {\sc MixDemAlloc} will be evaluated in the next section.

\section{Evaluation}
We provided analysis on the approximations ratios of our algorithms in the previous sections, which are the worst-case guarantees. In this section, we evaluate the empirical average-case ratios by simulations. We observe that our algorithms perform relatively well in several scenarios which are far below the theoretical worst-case values.

\subsection{Simulation Settings}

{\color{black} We consider two electric networks: a 38-node system adopted from \cite{singh2007loadmodel} (see Fig.~\ref{fig:net}), and the de-facto IEEE 123-node system.
\ifsupplementary
For the 38-node system, the settings of line impedance and maximum capacity are provided in the appendix.
\fi 
In 38-node system, we assume that the generation source is attached to node $1$, whereas the power demands are randomly generated at other 37 nodes uniformly.

 The IEEE 123-node networks are unbalanced three-phase networks with several devices that are not modeled in our formulation (Cons.~\raf{eq:c0}-\raf{clast}). As in \cite{gan2015exact}, we modify the IEEE network by the following:
\begin{itemize}
	\item The three phases are assumed to be decoupled into three identical single phase networks.
	\item Closed circuit switches are modelled as shorted lines and ignore open circuit switches.
	\item Transformers are modelled as lines with appropriate impedances.
\end{itemize}
We assume that the generation source is attached to the substation (node 150), whereas the power demands are randomly generated at the other nodes uniformly.
}

\begin{figure}[!htb]
	\begin{center}
		\includegraphics[scale=.7]{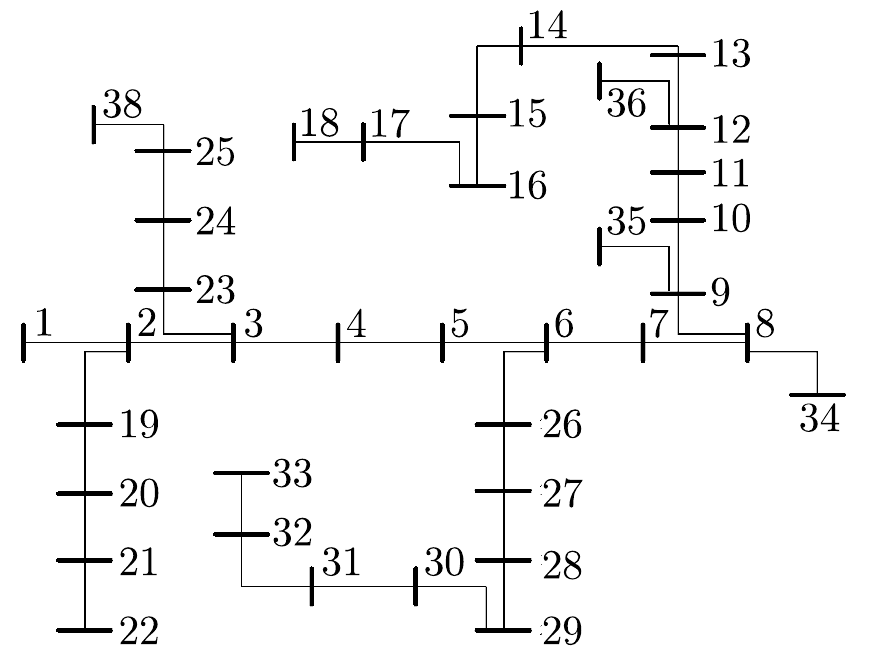}
	\end{center}
	\caption{A 38-node electric network for evaluation.}
	\label{fig:net}
\end{figure}

We consider diverse case studies of various settings of power demands by taking into account the correlation between customer demand and utility considering various demand types. The following are the settings of power demands at the customers:

\begin{enumerate}
	
	\item[(i)] {\em Utility-demand correlation}:
	\begin{enumerate}
		
		\item {\em Correlated setting (C)}: The utility of each customer is a function of the power demand:
		\begin{equation}\label{eq:valuationfunction}
		u_k({|s_k|}) = a\cdot {|s_k|}^2 + b\cdot {|s_k|} + c 
		\end{equation}
		where $a > 0, b, c \ge 0$ are constants. For simplicity, we consider $u_k(|s_k|) = |s_k|^2$.
		
		\item {\em Uncorrelated setting (U)}: The utility of each customer is independent of the power demand and is generated randomly from $[0, |s_{\max}(k)|]$. Here $s_{\max}(k)$ depends on the customer type (as defined the following). If customer $k$ is an industrial customer then $|s_{\max}(k)| = 1$MVA, otherwise $|s_{\max}(k)| = 5$KVA.
		
	\end{enumerate}
	\medskip
	
	\item[(ii)] {\em Customer types}:
	\begin{enumerate}
		
		\item {\em Residential (R) customers}: The customers are comprised of residential customers having small power demands ranging from 500VA to 5KVA.
		
		\item {\em Industrial (I) customers}: The customers have big demands ranging from 300KVA up to 1MVA and non-negative reactive power.
		
		\item {\em Mixed (M) customers}: The customers are comprised of a mix of industrial and residential customers.  Industrial customers constitute no more than 20\% of all customers chosen at random.

	\end{enumerate}
	
\end{enumerate}

In this paper, the case studies will be represented by the aforementioned acronyms. For example, the case study named CM stands for the one with mixed customers and correlated utility-demand setting.

In order to quantify the performance of our algorithms, we use Gurobi  optimizer to obtain numerically  close-to-optimal solutions for {\sc MaxOPF} and {\sc sMaxOPF} respectively.
 We denote output solution for {\sc MaxOPF} (resp., {\sc sMaxOPF}) obtained by Gurobi optimizer  by $\OPT$ (resp., $\OPT_s$). 
 To ensure the feasibility of $\OPT_s$, we perform a linear search similar to that in Algorithm~\ref{alg:lin} with a small modification.
Note that there is no guarantee that the optimizer will return an optimal solution nor it will terminate in a reasonable time (e.g., within 500 seconds for each run). Whenever the optimizer exceeds the time limit, the current best solution is considered to be optimal.

 We set the step size to be $\epsilon=0.005$ (i.e., $0.5\%$) for both {\sc MixDemAlloc} and $\OPT_s$. The power factor for each customer varies between $0.8$ to $1$ (to comply with IEEE standards) and thus we restrict the phase angle $\theta$ of demands to be in the range of $[-36^{\circ}, 36^{\circ}]$.

The simulations were evaluated using 2 Quad core Intel Xeon CPU E5607 2.27 GHz processors with 12 GB of RAM. The algorithms were implemented using Python programming language with Scipy library for scientific computation. 

\subsection{Evaluation Results}

\begin{figure*}[!htb]
	\begin{subfigure}[b]{0.5\textwidth}
		\begin{center}
			\includegraphics[scale=0.5]{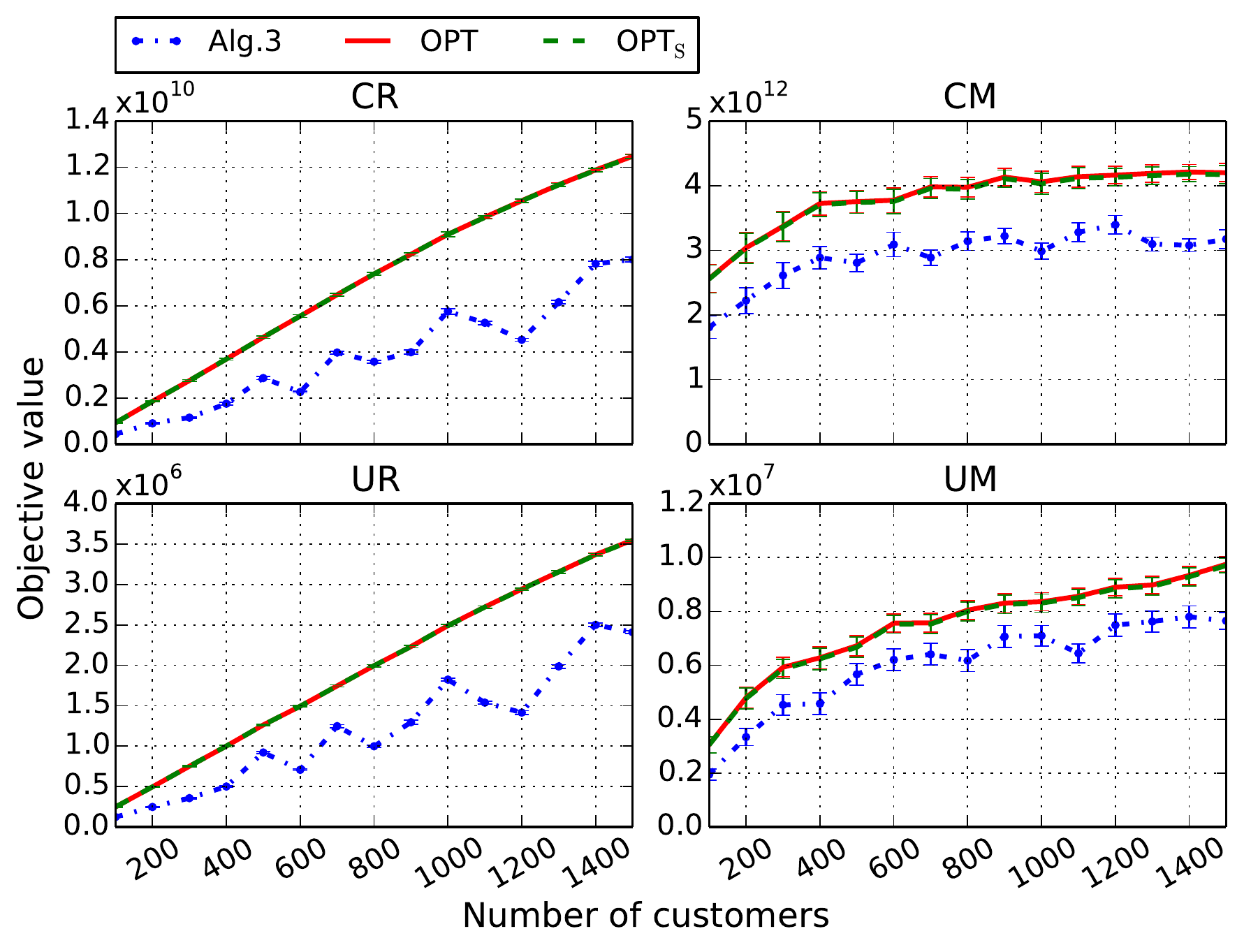}
		\end{center}
		\caption{}
		\label{fig:obj38}
	\end{subfigure}
	\begin{subfigure}[b]{0.5\textwidth}
		\begin{center}
			\includegraphics[scale=0.5]{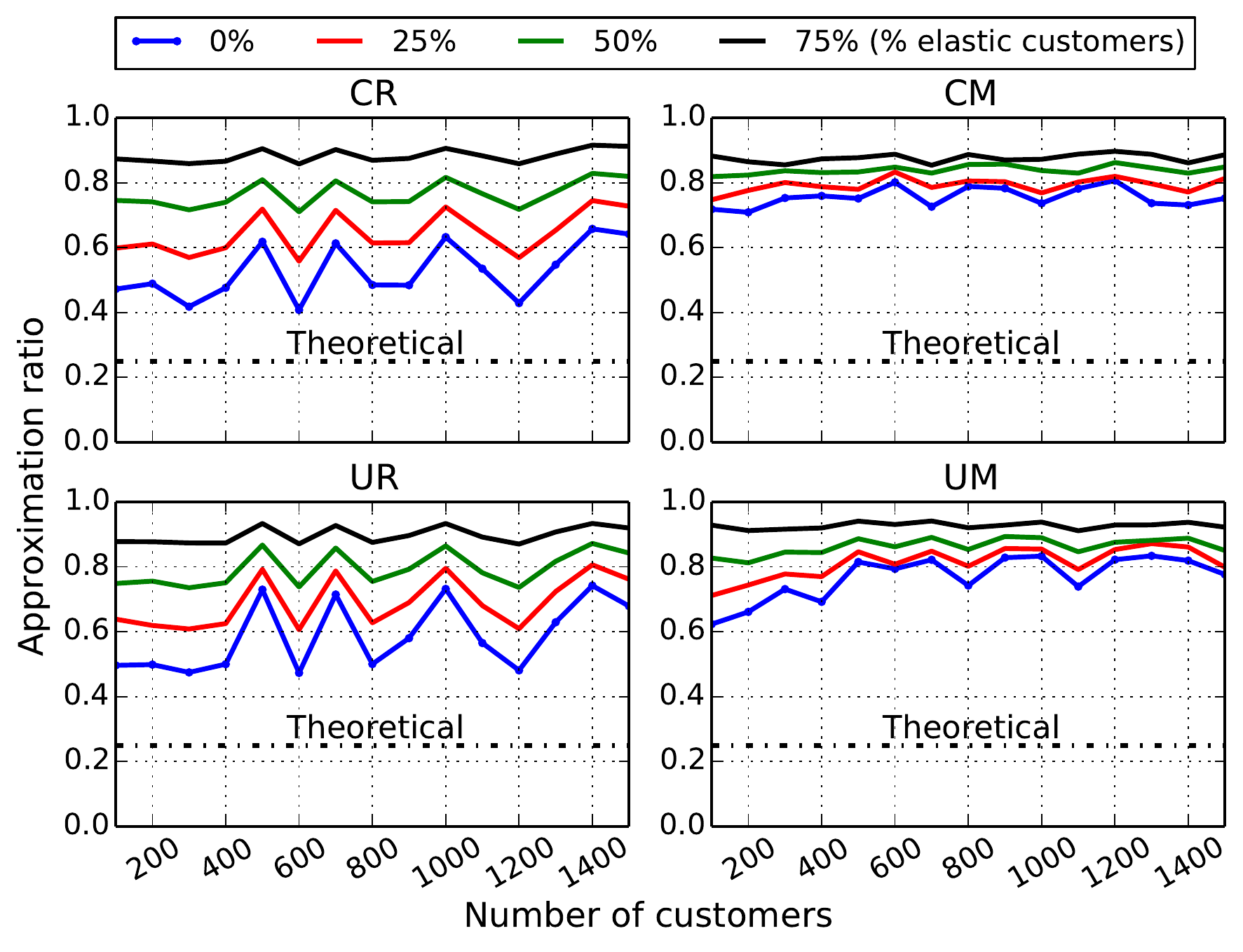}
		\end{center}
		\caption{}
		\label{fig:ar38}
	\end{subfigure}
	\caption{\color{black} The 38-node system: (a) The average objective values of {\sc MixDemAlloc}  with inelastic demands only (denoted by {\sf Alg.3}), $\OPT$ and {\sc $\OPT_s$}. (b) The average approximation ratios of {\sc InelasDemAlloc} applied to instances with different percentage of elastic demands, against the number of customers with 95\% confidence interval. }
\end{figure*}
\begin{figure*}[!htb]
	\begin{subfigure}[b]{0.5\textwidth}
		\begin{center}
			\includegraphics[scale=0.5]{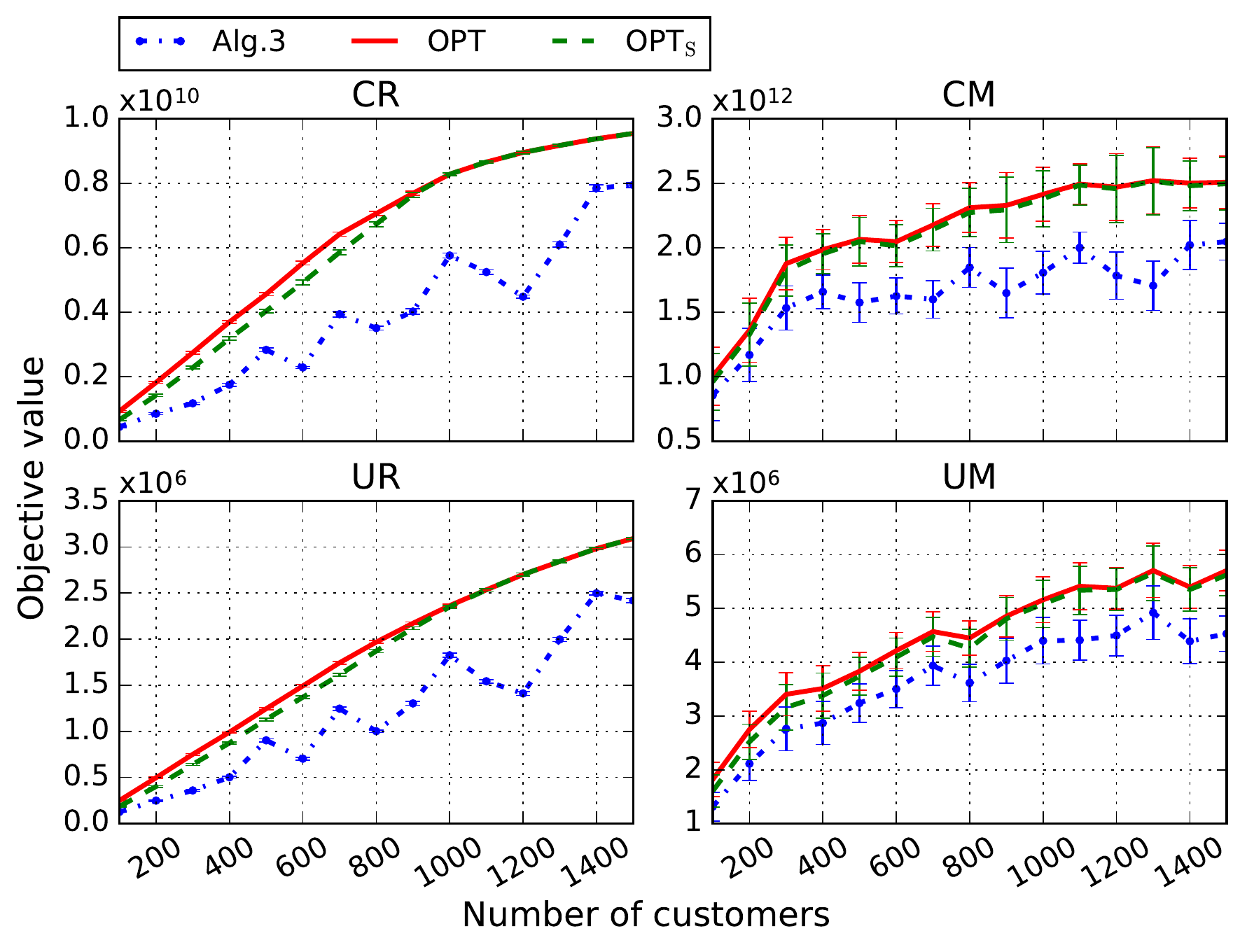}
		\end{center}
		\caption{}
		\label{fig:obj123}
	\end{subfigure}
	\begin{subfigure}[b]{0.5\textwidth}
		\begin{center}
			\includegraphics[scale=0.5]{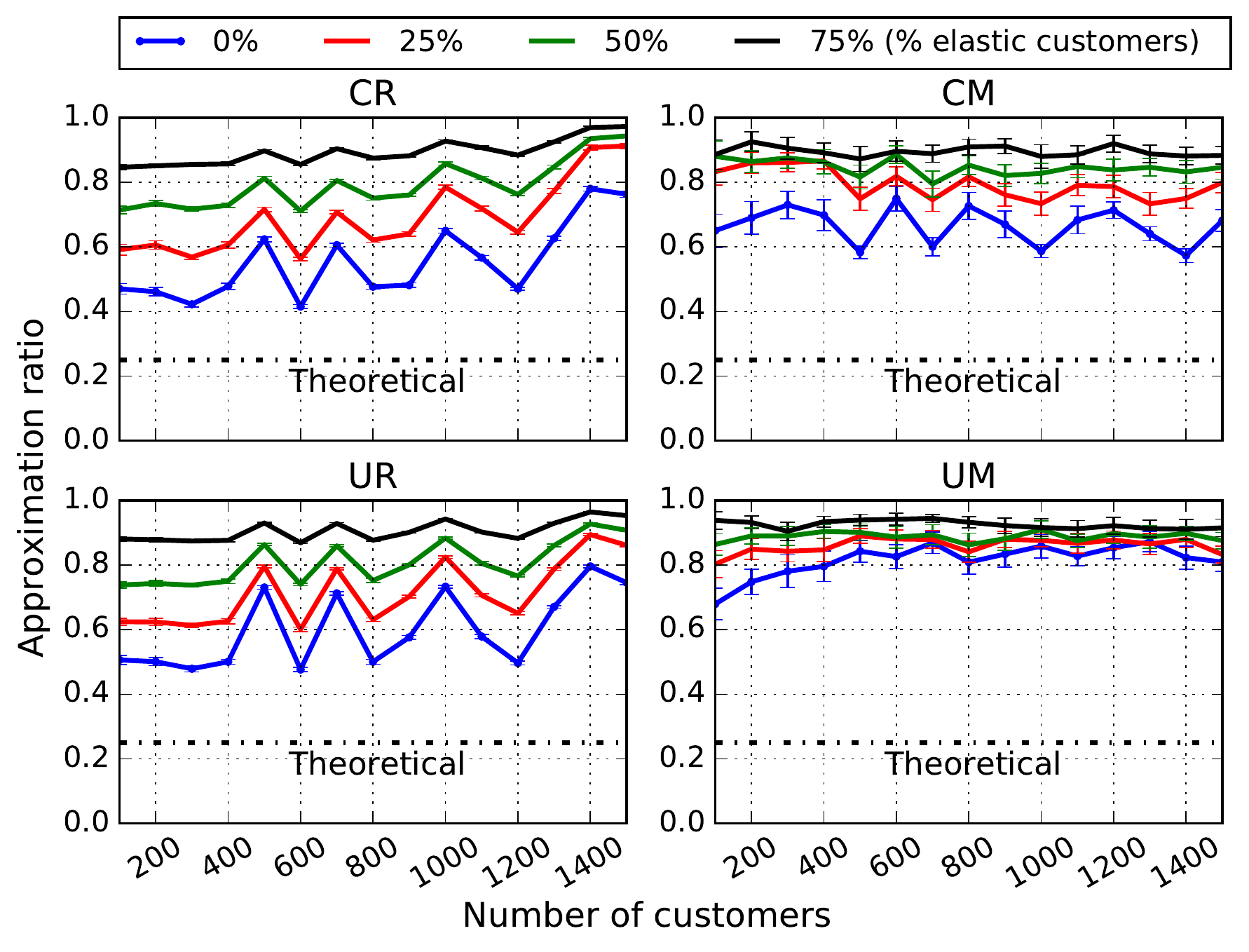}
		\end{center}
		\caption{}
		\label{fig:ar123}
	\end{subfigure}
		\caption{\color{black}  The IEEE 123-node system: (a) The average objective values of {\sc MixDemAlloc}  with inelastic demands only (denoted by {\sf Alg.3}), $\OPT$ and {\sc $\OPT_s$}. (b) The average approximation ratios of {\sc InelasDemAlloc} applied to instances with different percentage of elastic demands, against the number of customers with 95\% confidence interval.}
	\label{fig:sim}
\end{figure*}
\begin{figure*}[!htb]
	\begin{subfigure}[b]{0.5\textwidth}
		\begin{center}
			\includegraphics[scale=0.5]{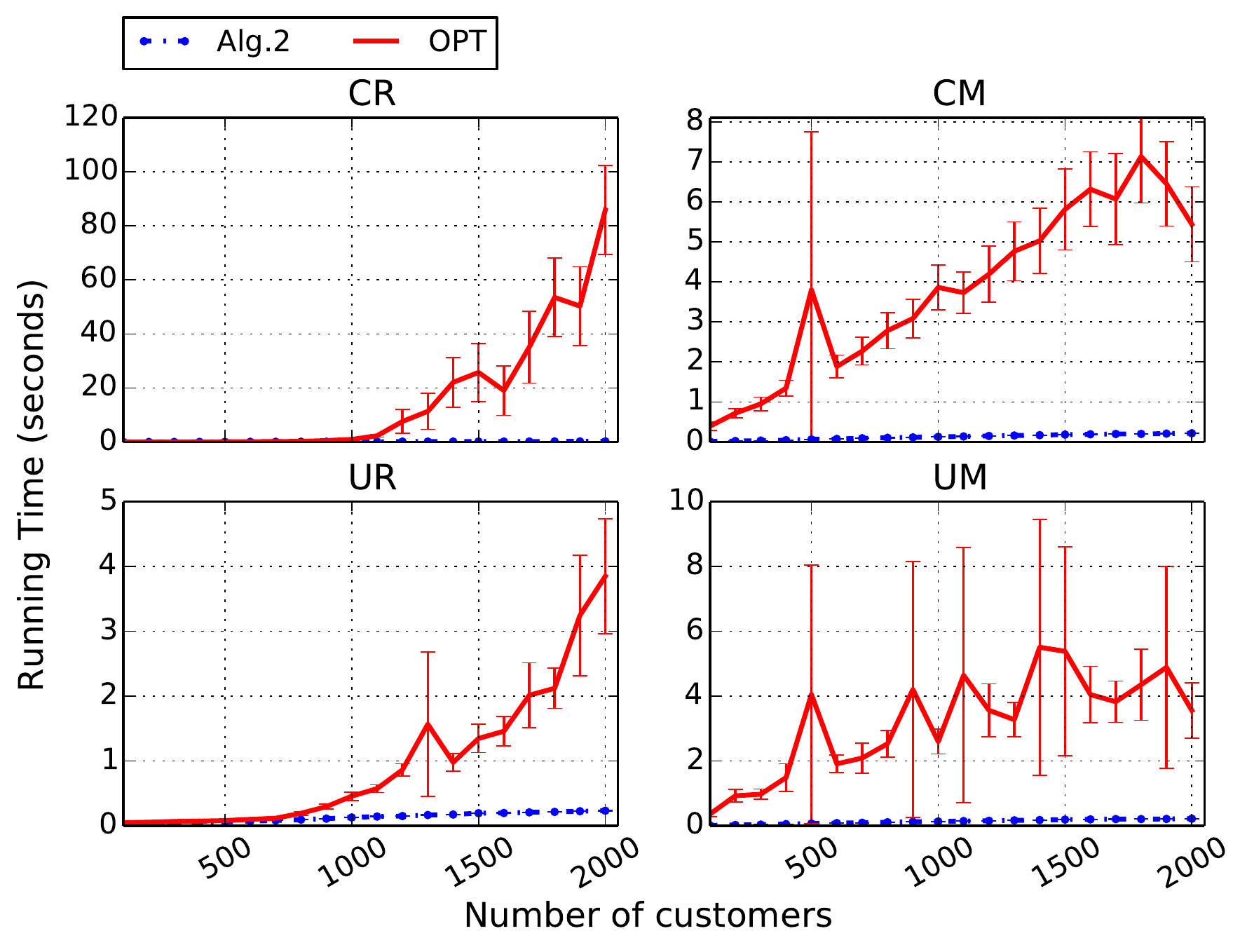}
		\end{center}
		\caption{}
	\end{subfigure}
	\begin{subfigure}[b]{0.5\textwidth}
		\begin{center}
			\includegraphics[scale=0.5]{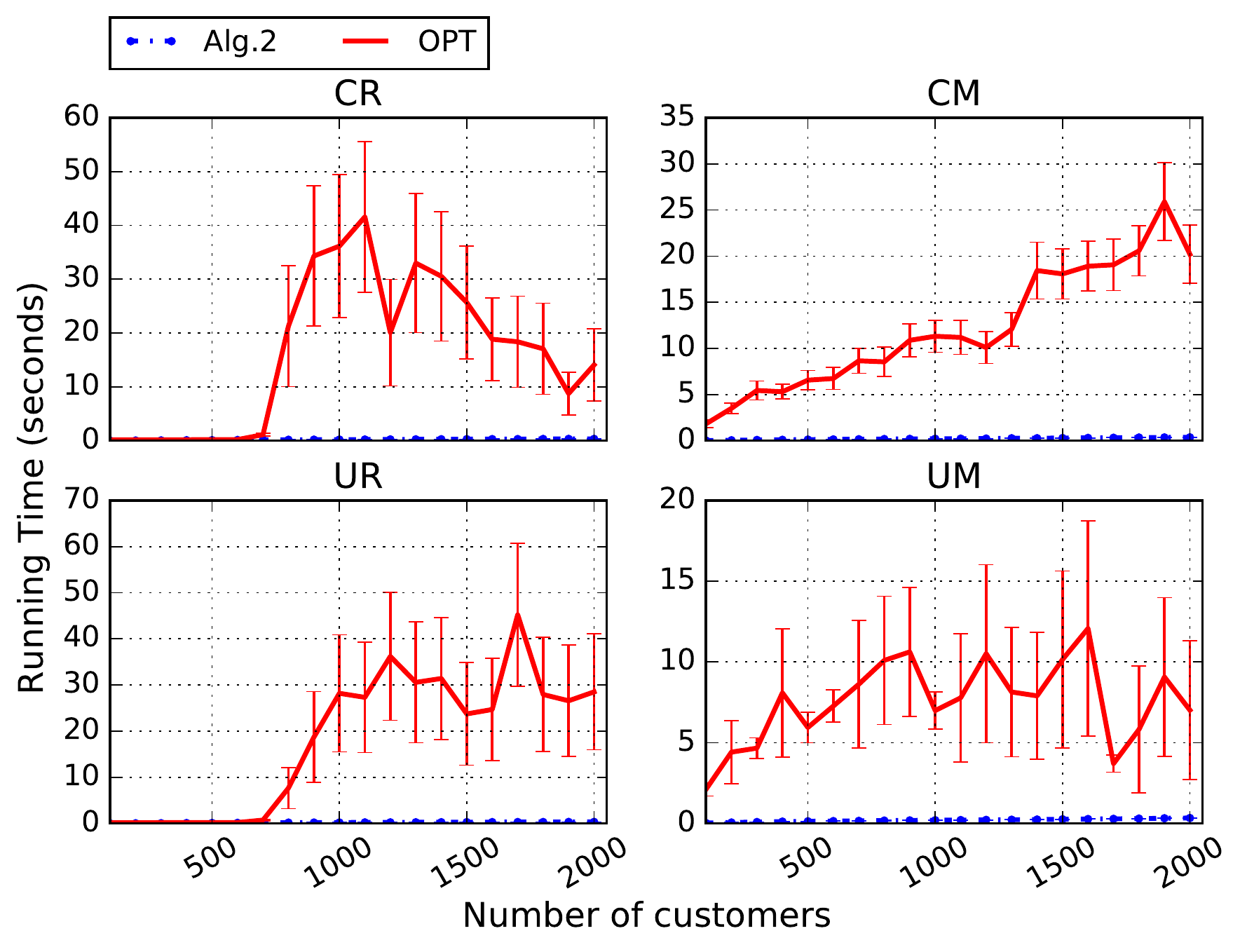}
		\end{center}
		\caption{}
	\end{subfigure}
	\caption{\color{black} The average running time of {\sc InelasDemAlloc} (denoted by {\sf Alg.2}) and $\OPT$ for (a) the 38-node system and (b) the IEEE 123-node system against the number of customers with 95\% confidence interval. }
	\label{fig:time}
\end{figure*}

\subsubsection{Optimality}
\

{\color{black} Fig~\ref{fig:obj38} (resp., \ref{fig:obj123}) present the objective value attained by {\sc MixDemAlloc} with only inelastic demands, $\OPT$, and $\OPT_s$ respectively using the 38-node system (resp., the IEEE 123-node system) for up to $1500$ customers. }Each run is repeated $40$ times.
The utility values attained by $\OPT$ and $\OPT_s$ are almost identical in all scenarios. This is due to the insignificance of the terms associated with transmission power loss in {\sc MaxOPF}. We observe from the figure that {\sc MixDemAlloc} performs relatively better when loads are mixed between residential and industrial (CM and UM). 

We note that {\sc MixDemAlloc} objective does not smoothly increase in the number of customers which is due to the way customers are arranged into different groups in algorithm {\sc InelasDemAlloc}. Customer utility is rounded by the factor $L$ which is a function of the number of customers. We observe from the figure that such rounding sometimes obtains lower utilities by increasing the number of customers.
 
{\color{black} The empirical approximation ratios for the two networks are plotted in Fig.~\ref{fig:ar38} and \ref{fig:ar123} against the number of customers, along with the theoretical approximation ratio given by Theorem~\ref{thm:greedyalloc} part 1.} The lines in Fig.~\ref{fig:ar38} (resp.,~\ref{fig:ar123}) correspond to different percentages of elastic demands (i.e.,  $\tfrac{|\cF|}{|\cN|}=0, 0.25, 0.50, 0.75$).  When a line is close to $y=1$, it is close to the optimal solution. As the percentage of elastic demands increases, {\sc MixDemAlloc} consistently achieves better solutions in all scenarios. The average empirical ratios are more than $0.4$ in all cases which is well above the theoretical worst case results. This suggests that {\sc MixDemAlloc} performs relatively well in practice under difference scenarios.
 

\medskip

\subsubsection{Transmission Power Loss}
\

To understand the transmission power loss in practice, we evaluate the loss ratio (i.e., $\delta$) in {\sc MixDemAlloc} (with inelastic demands only) and $\OPT_s$ respectively for the 38-node system. The results are plotted in Fig.~\ref{fig:loss}. As one may expect, $\OPT_s$ has a higher loss percentage since it satisfies more demands than {\sc MixDemAlloc} in general. We observe that when customers are all residential, {\sc MixDemAlloc} always obtains feasible solutions without any reduction in link capacities (i.e., $\delta= 0$). The maximum loss ratio obtained is $5.5\%$ in UM scenario for both $\OPT_s$ and {\sc MixDemAlloc}. The ramification is that {\sc MixDemAlloc} can attain a good empirical approximation ratio in practice, because the transmission power loss is usually small in practical electric networks.

\begin{figure}[!h]
	\begin{center}\hspace*{-5pt}
		\includegraphics[scale=0.5]{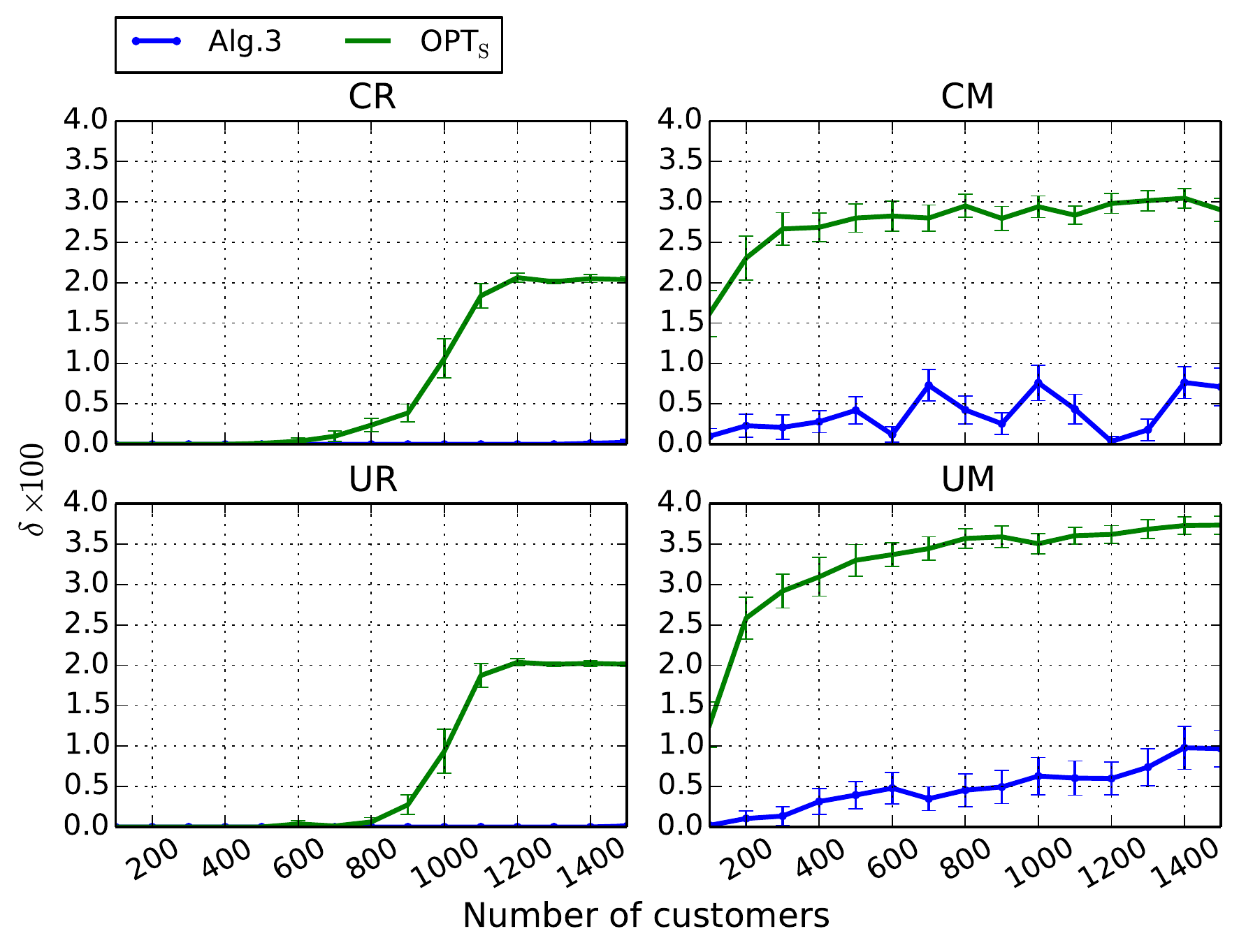}
	\end{center}
	\caption{The average loss of {\sc MixDemAlloc}  with inelastic demands only (denoted by {\sf Alg.3}) and $\OPT_s$ for the $38$ node system.}
	\label{fig:loss} 
\end{figure}

\medskip

\subsubsection{Running Time}
\

One of the main goals of this work is to develop efficient algorithms that ensure a polynomial running time.
The computational time of {\sc InelasDemAlloc} is compared against the Gurobi solver. Computational time is of significant importance when designing centralized controllers for micro-grids since this will have implications on the overall stability. 
The running time is presented in Fig.~\ref{fig:time} under different scenarios for up to $2000$ customers, each point is repeated $100$ times.  

We observe the running time of {\sc InelasDemAlloc} is always in milliseconds and linearly increases in the number of customers $n$. On the other hand, the average running time of $\OPT$ is much higher in many cases (measured in minuets) and has no polynomial guarantee. Throughout the simulations, we observed many timeouts especially in scenario CR. The actual running time of $\OPT$ may substantially increase if we increase the timeout parameter in Gurobi optimizer. 
The running time of  $\OPT$ can be much higher if we consider larger network topologies, whereas, linear increase is expected for {\sc InelasDemAlloc} in practice. Therefore, our algorithm is far more scalable than any known optimal algorithm.
We note that the implementation of our algorithms can be further optimized using C programming language since the current one is based on Python that is relatively slow.


\section{Conclusion}\label{sec:concl}

While optimal power flow problem has been extensively considered in power engineering literature, the theoretical understanding of this problem is lacking. Recent advances in convex relaxation techniques for optimal power flow problem \cite{gan2015exact,low2014convex1,low2014convex2} have generated substantial leaps in proven efficient algorithms for optimal power flow problem, which thus far were only applied to demand response with elastic demands. 

In order to advance the frontier for tackling optimal power flow problem, we consider combinatorial allocation of inelastic demands considering power flows. We first showed the hardness of this problem in a general form. We next presented an efficient approximation algorithm to a relaxed problem. Our simulation studies show that the proposed algorithm can produce close-to-optimal solutions in practice. Our results present the first step of fundamental understanding of combinatorial allocation problem of power flows, which naturally extends the classical real-valued combinatorial flow optimization, but is also a substantial departure from the classical problem.  Our results generalize the recent works of complex-demand knapsack problem and unsplittable flow problem.  Recently, power allocation has been extended to consider scheduling problems \cite{KKEC16CSP,KKCE16b}.

\bibliographystyle{ieeetr}
\bibliography{reference}

\begin{thebibliography}{10}

\bibitem{gan2015exact}
L.~Gan, N.~Li, U.~Topcu, and S.~H. Low, ``Exact convex relaxation of optimal
  power flow in radial networks,'' {\em IEEE Transactions on Automatic
  Control}, vol.~60, no.~1, pp.~72--87, 2015.

\bibitem{low2014convex1}
S.~Low, ``Convex relaxation of optimal power flow, part {I}: Formulations and
  equivalence,'' {\em IEEE Transactions on Control of Network Systems}, vol.~1,
  pp.~15--27, March 2014.

\bibitem{low2014convex2}
S.~Low, ``Convex relaxation of optimal power flow, part {II}: Exactness,'' {\em
  IEEE Transactions on Control of Network Systems}, vol.~1, pp.~177--189, June
  2014.

\bibitem{li2012optimization}
N.~Li, L.~Gan, L.~Chen, and S.~H. Low, ``An optimization-based demand response
  in radial distribution networks,'' in {\em Globecom Workshops (GC Wkshps),
  2012 IEEE}, pp.~1474--1479, IEEE, 2012.

\bibitem{CKM14}
C.-K. Chau, K.~Elbassioni, and M.~Khonji, ``Truthful mechanisms for
  combinatorial {AC} electric power allocation,'' in {\em International
  Conference on Autonomous Agents and Multiagent Systems (AAMAS)}, 2014.
\newblock {http://arxiv.org/abs/1403.3907}.

\bibitem{chekuri2009unsplittable}
C.~Chekuri, A.~Ene, and N.~Korula, ``Unsplittable flow in paths and trees and
  column-restricted packing integer programs,'' in {\em Approximation,
  Randomization, and Combinatorial Optimization. Algorithms and Techniques},
  pp.~42--55, Springer, 2009.

\bibitem{anagnostopoulos2014mazing}
A.~Anagnostopoulos, F.~Grandoni, S.~Leonardi, and A.~Wiese, ``A mazing 2+
  $\varepsilon$ approximation for unsplittable flow on a path,'' in {\em SODA},
  pp.~26--41, SIAM, 2014.

\bibitem{baran1989placement}
M.~E. Baran and F.~F. Wu, ``Optimal capacitor placement on radial distribution
  systems,'' {\em IEEE Transactions on Power Delivery}, vol.~4, no.~1,
  pp.~725--734, 1989.

\bibitem{baran1989sizing}
M.~E. Baran and F.~F. Wu, ``Optimal sizing of capacitors placed on a radial
  distribution system,'' {\em IEEE Transactions on Power Delivery}, vol.~4,
  no.~1, pp.~735--743, 1989.

\bibitem{equal}
B.~Subhonmesh, S.~Low, and K.~Chandy, ``Equivalence of branch flow and bus
  injection models,'' in {\em Communication, Control, and Computing (Allerton),
  2012 50th Annual Allerton Conference on}, pp.~1893--1899, Oct 2012.

\bibitem{farivar2013branch}
M.~Farivar and S.~H. Low, ``Branch flow model: Relaxations and convexification
  -- part i,'' {\em IEEE Transactions on Power Systems}, vol.~28, no.~3,
  pp.~2554--2564, 2013.

\bibitem{YC13CKP}
L.~Yu and C.-K. Chau, ``Complex-demand knapsack problems and incentives in {AC}
  power systems,'' in {\em International Conference on Autonomous Agents and
  Multiagent Systems (AAMAS)}, 2013.
\newblock {http://arxiv.org/abs/1205.2285}.

\bibitem{woeginger2000does}
G.~J. Woeginger, ``When does a dynamic programming formulation guarantee the
  existence of a fully polynomial time approximation scheme (fptas)?,'' {\em
  INFORMS Journal on Computing}, vol.~12, no.~1, pp.~57--74, 2000.

\bibitem{CKM15}
C.-K. Chau, K.~Elbassioni, and M.~Khonji, ``Truthful mechanisms for
  combinatorial allocation of electric power in alternating current electric
  systems for smart grid,'' {\em ACM Transactions on Economics and
  Computation}, vol.~5, pp.~7:1--7:29, Oct 2016.
\newblock http://arxiv.org/abs/1507.01762.

\bibitem{KKCEZ16}
A.~Karapetyan, M.~Khonji, C.-K. Chau, K.~Elbassioni, and H.~Zeineldin,
  ``Efficient algorithm for scalable event-based demand response management in
  microgrids,'' {\em to appear in IEEE Transactions on Smart Grid}, 2016.
\newblock {http://arxiv.org/abs/1610.03002}.

\bibitem{KCE14}
M.~Khonji, C.-K. Chau, and K.~M. Elbassioni, ``Inapproximability of power
  allocation with inelastic demands in {AC} electric systems and networks,'' in
  {\em International Workshop on Smart Complex Engineered Networks (SCENE
  2014); a workshop within the 23rd International Conference on Computer
  Communication and Networks {ICCCN})}, pp.~1--6, 2014.

\bibitem{KR72}
R.~Karp, ``Reducibility among combinatorial problems,'' in {\em Complexity of
  Computer Computations}, The IBM Research Symposia Series, pp.~85--103,
  Springer US, 1972.

\bibitem{KM96}
J.~M. Kleinberg, {\em Approximation algorithms for disjoint paths problems}.
\newblock PhD thesis, 1996.

\bibitem{VAV04}
K.~Varadarajan and G.~Venkataraman, ``Graph decomposition and a greedy
  algorithm for edge-disjoint paths,'' in {\em SODA}, pp.~379--380, 2004.

\bibitem{guruswami2003near}
V.~Guruswami, S.~Khanna, R.~Rajaraman, B.~Shepherd, and M.~Yannakakis,
  ``Near-optimal hardness results and approximation algorithms for
  edge-disjoint paths and related problems,'' {\em Journal of Computer and
  System Sciences}, vol.~67, no.~3, pp.~473--496, 2003.

\bibitem{AYR01}
Y.~Azar and O.~Regev, ``Strongly polynomial algorithms for the unsplittable
  flow problem,'' in {\em Integer Programming and Combinatorial Optimization}
  (K.~Aardal and B.~Gerards, eds.), vol.~2081 of {\em Lecture Notes in Computer
  Science}, pp.~15--29, Springer Berlin Heidelberg, 2001.

\bibitem{CKS06}
C.~Chekuri, S.~Khanna, and F.~B. Shepherd, ``An o ($\sqrt{n}$) approximation
  and integrality gap for disjoint paths and unsplittable flow,'' {\em Theory
  of computing}, vol.~2, no.~7, pp.~137--146, 2006.

\bibitem{AMC05}
M.~Andrews, J.~Chuzhoy, S.~Khanna, and L.~Zhang, ``Hardness of the undirected
  edge-disjoint paths problem with congestion,'' in {\em FOCS}, pp.~226--241,
  IEEE, 2005.

\bibitem{garg1997primal}
N.~Garg, V.~V. Vazirani, and M.~Yannakakis, ``Primal-dual approximation
  algorithms for integral flow and multicut in trees,'' {\em Algorithmica},
  vol.~18, no.~1, pp.~3--20, 1997.

\bibitem{LGH16AChard}
K.~Lehmann, A.~Grastien, and P.~V. Hentenryck, ``{AC}-feasibility on tree
  networks is {NP}-hard,'' {\em IEEE Transactions on Power Systems}, vol.~31,
  pp.~798--801, Jan 2016.

\bibitem{verma09thesis}
A.~Verma, ``Power grid security analysis: An optimization approach,'' tech.
  rep., Columbia University, 2009.
\newblock PhD diss.

\bibitem{bv15AChard}
D.~Bienstock and A.~Verma, ``Strong np-hardness of ac power flows
  feasibility,'' tech. rep., 2015.
\newblock {http://arxiv.org/abs/1512.07315}.

\bibitem{singh2007loadmodel}
D.~Singh, R.~Misra, and D.~Singh, ``Effect of load models in distributed
  generation planning,'' {\em IEEE Transactions on Power Systems}, vol.~22,
  pp.~2204--2212, Nov 2007.

\bibitem{KKEC16CSP}
M.~Khonji, A.~Karapetyan, K.~Elbassioni, and C.-K. Chau, ``Complex-demand
  scheduling problem with application in smart grid,'' in {\em International
  Computing and Combinatorics Conference (COCOON)}, 2016.
\newblock http://arxiv.org/abs/1603.01786.

\bibitem{KKCE16b}
A.~Karapetyan, M.~Khonji, C.-K. Chau, and K.~Elbassioni, ``Online algorithm for
  demand response with inelastic demands and apparent power constraint,'' tech.
  rep., Masdar Institute, 2016.
\newblock https://arxiv.org/abs/1611.00559.

\end{thebibliography}

\appendix

\subsection{Derivation of Branch Flow Model for Trees}
The branch flow model can be derived from Eqns.~\raf{eq:1}-\raf{eq:3} as follows. 
First rewrite Eqn.~\raf{eq:1} by taking the complex conjugate of the both sides:
\begin{align}
& I_{i,j} = \frac{S_{i,j}^\ast}{V_i^\ast} \label{eq:i}\\
\Rightarrow \ \ & |I_{i,j} |^2 =  \frac{|S_{i,j}|^2}{|V_i|^2} \label{eq:ii}
\end{align}
Substituting Eqn.~\raf{eq:i} in Eqn.~\raf{eq:2}, we obtain:
$$
V_j = V_i - I_{i,j} z_{i,j}= V_i - \frac{S_{i,j}^\ast}{V_i^\ast} z_{i,j}
$$
Taking the magnitude square of the both sides yields:
\begin{align}
|V_j|^2 &=  |V_i|^2 + | \tfrac{S_{i,j}^\ast}{V_i^\ast} z_{i,j} |^2 -\notag\\
&\qquad \qquad  2 |V_i| |\tfrac{S_{i,j}^\ast}{V_i^\ast} z_{i,j}| (\sin \theta_1 \sin \theta_2 + \cos \theta_1 \cos \theta_2) \notag\\
&=  |V_i|^2 + \frac{|S_{i,j}|^2}{|V_i|^2} |z_{i,j}|^2 - 2 |V_i| |\tfrac{S_{i,j}^\ast}{V_i^\ast} z_{i,j}| \cos( \theta_1 - \theta_2), \label{eq:b1}
\end{align}
where we let $\theta_1 \triangleq \arg(V_i)$ and $\theta_2 \triangleq \arg(\tfrac{S_{i,j}^\ast}{V_i^\ast} z_{i,j})$. 
Note that 
$$
\re\Big(\frac{V_i}{\tfrac{S_{i,j}^\ast}{V_i^\ast} z_{i,j}} \Big) = \frac{|V_i|}{\tfrac{|S_{i,j}^\ast|}{|V_i^\ast|}| z_{i,j}|}  \cos( \theta_1 - \theta_2)
$$
We can then rewrite Eqn.~\raf{eq:b1} in the following form
\begin{align*}
|V_j|^2 & =  |V_i|^2 + \frac{|S_{i,j}|^2}{|V_i|^2} |z_{i,j}|^2 -2   |\tfrac{S_{i,j}^\ast}{V_i^\ast} z_{i,j}|^2 \re\Big(\frac{V_i}{\tfrac{S_{i,j}^\ast}{V_i^\ast} z_{i,j}} \Big) \\
 & = |V_i|^2 + \frac{|S_{i,j}|^2}{|V_i|^2} |z_{i,j}|^2 -2   \re\Big(	\frac{|S_{i,j}|^2}{|V_i|^2} |z_{i,j}|^2  \frac{V_i \cdot V_i^\ast}{S_{i,j}^\ast z_{i,j}} \Big)
\end{align*}
Then applying the following from the properties of complex numbers: $|V_i|^2 = V_i \cdot V_i^\ast$, $\tfrac{1}{z_{i,j}}=\frac{z^\ast_{i,j}}{|z_{i,j}|^2} $, and $\tfrac{1}{S^\ast_{i,j}}=\frac{S_{i,j}}{|S_{i,j}|^2} $ we obtain:
\begin{align}
|V_j|^2 & = |V_i|^2 + \frac{|S_{i,j}|^2}{|V_i|^2} |z_{i,j}|^2 -2 \re\Big( \frac{|S_{i,j}|^2}{|V_i|^2} |z_{i,j}|^2  \frac{|V_i|^2 z_{i,j}^\ast S_{i,j}}{|S_{i,j}|^2 |z_{i,j}|^2}	\Big) \notag\\
& = |V_i|^2 + \frac{|S_{i,j}|^2}{|V_i|^2} |z_{i,j}|^2 -2 \re(z_{i,j}^\ast \cdot S_{i,j}) \label{eq:v}
\end{align}
 The branch follow model is obtained by Eqns.~\raf{eq:3}, \raf{eq:ii}, \raf{eq:v}.

\subsection{Proofs}

In the following proofs, we rely on a hardness result of a well-known (weakly) NP-Hard problem called the Subset Sum problem (\textsc{SubSum}).
\medskip

\begin{definition}[\textsc{SubSum}]
Given a set of positive integers ${A} \triangleq\{a_1,\ldots,a_m\}$ and a positive integer $B$, decide if there exists a subset of ${A}$ that sums-up to exactly $B$.
\end{definition}
Note that $B$ is generally not polynomial in $m$. Otherwise, \textsc{SubSum} can be solved easily in polynomial time by dynamic programming.

\medskip

\begin{customthm}{1}
	Unless {\rm P}={\rm NP}, there is no $(\alpha,\beta)$-approximation for {\sc MaxOPF$_{\sf V}$} (even when $|\cE| = 1$) by a polynomial-time algorithm in $n$, for any $\alpha$ and $\beta$ have polynomial length in $n$.
\end{customthm}
\begin{proof} 
The basic idea is that we show a reduction from \textsc{SubSum} to \textsc{MaxOPF}$_{\sf V}$. 
Assume that there is an ($\alpha, \beta$)-approximation for \textsc{MaxOPF}$_{\sf V}$. We construct an instance $I'$ of  \textsc{MaxOPF}$_{\sf V}$ for each instance  $I$ of \textsc{SubSum}, such that \textsc{SubSum}$(I)$ is a ``yes'' instance if and only if the $(\alpha, \beta)$-approximation of 	\textsc{MaxOPF}$_{\sf V}(I')$ gives a total utility at least $\alpha$. Since \textsc{SubSum} is NP-hard, there exists no $(\alpha,\beta)$-approximation for {\sc MaxOPF$_{\sf V}$} in polynomial time. Otherwise, \textsc{SubSum} can be solved in polynomial time. 

Given a \textsc{SubSum} instance $I = (A, B)$, where $A = \{a_1,...,a_m \}$, we define \textsc{MaxOPF}$_{\sf V}$ instance $I'$ as follows.

\begin{itemize}

\item Consider a network with a single edge $e=(0,1)$. Let $z_e \triangleq 1 + \bf i$. 
\item Fix some $\delta >0$, set $v_0$ in instance $I'$ by
$$
v_0 \triangleq \frac{1}{\tfrac{\delta}{2} - \epsilon}  \Big|\sum_{k=1}^{m+1} s_k\Big|^2, \ \text{ for arbitrarily small } \epsilon > 0.
$$ 

Set $v_{\max} = v_0 + \delta$ and $v_{\min} = v_0 - \delta$.
\item Let $\cN= \cI =  \{1,...,m+1\}$ be the set of customers attached to node $1$ (see Fig.~\ref{fig:hard-V}). 
Define  
\begin{align*}
\overline\Lambda &\triangleq  \max\{v_0 - \tfrac{1}{\beta} v_{\min}, \beta v_{\max} - v_0\}. 
\end{align*}
For each $k \in \{1,...,m+1\}$, define the customers' demands and utilities as follows:
$$
s_k  \triangleq \overline\Lambda a_k, \quad   u_k \triangleq  \tfrac{\alpha}{m+1}	 \quad
s_{m+1} \triangleq - {\bf i}  \overline\Lambda B, \quad  u_{m+1} \triangleq 1
$$

\end{itemize}

\begin{figure}[!htb]
	\begin{center}
		\includegraphics{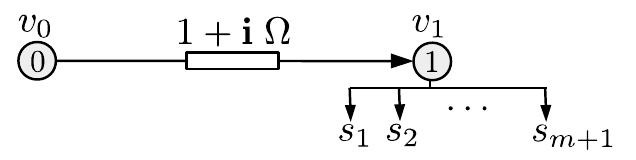}
	\end{center}
	\caption{A gadget for reduction from \textsc{SubSum} to \textsc{MaxOPF}$_{\sf v}$.}
	\label{fig:hard-V}
\end{figure}

First, we prove that if \textsc{SubSum}$(I)$ is a ``yes'' instance, then the $(\alpha, \beta)$-approximation of \textsc{MaxOPF}$_{\sf V}(I')$ gives a total utility at least $\alpha$. If \textsc{SubSum}$(I)$ is a ``yes'' instance, then $\sum_{k = 1}^{m}a_k \hat{x}_k = B$, where $\hat x \in \{0,1\}^m$ is a solution of \textsc{SubSum}. We construct a solution $x\in \{0,1\}^{m+1}$ for \textsc{MaxOPF}$_{\sf V}(I')$ by
$$
x_k = \left\{
\begin{array}{l l}
\hat x_k & \quad \text{if $k = 1, ..., m$}\\
1 & \quad \text{if $k = m+1$}
\end{array} \right.
$$
We formulate Cons.~\raf{eq:c0} as $S_e = \sum_{k\in \cN} s_k + z_e \ell_e$ and substitute it in Cons.~\raf{eq:c2} to obtain:
$$
\tfrac{1}{2}(v_0 - v_{\max} ) \le  \re(z_{e }^\ast  S_{e }) - \tfrac{1}{2}|z_e|^2 \ell_e  \le \tfrac{1}{2}(v_0 - v_{\min})
$$
Note that
\begin{align*}
\tfrac{1}{2}(v_0 - v_1) = \ & \re(z_{e }^\ast  S_{e }) - \tfrac{1}{2}|z_e|^2 \ell_e  \\
= \ & \re( \sum_{k\in\cN} z_{e }^\ast s_k x_k + |z_e|^2 \ell_e ) - \tfrac{1}{2}|z_e|^2 \ell_e  \\
= \ &  \sum_{k=1}^{m+1} (z_e^{\rm R}s^{\rm R}_k + z_e^{\rm I}s^{\rm I}_k)x_k + \tfrac{1}{2}|z_e|^2 \ell_e \\
= \ &  \overline{\Lambda} \Big(\sum_{k=1}^m  a_k x_k -  B x_{m+1}\Big) +\ell_e   =  \ell_e,
\end{align*}
where $\sum_{k = 1}^{m}a_k \hat x_k - B = 0$ and $|z_e|^2 = 2$.

By Cons.~\raf{eq:c0} and the definition of $v_0$, we obtain: 
\begin{equation}
\ell_e = \frac{|S_e|^2}{v_0} = \frac{(\tfrac{\delta}{2}- \epsilon)\cdot \Big|\sum_{k=1}^{m+1} s_k x_k\Big|^2 }{ \Big|\sum_{k=1}^{m+1} s_k\Big|^2} \le  \tfrac{\delta}{2} - \epsilon \label{eq:le}
\end{equation}
Hence,
\begin{align*}
0 \le \frac{v_0 - v_1}{2} = \ell_e  \le  \tfrac{\delta}{2}- \epsilon <\frac{v_0 - v_{\min}}{2}
\end{align*}
Therefore, $v_{\min} \le v_1 \le v_0 = v_{\max}$ and Cons.~\raf{eq:c2} is satisfied.
Since $u_{m+1}=1$, we have $u(x) \ge 1$, and $\OPT$ is also at least 1. By the feasibility of this solution, the $(\alpha, \beta)$-approximation of \textsc{MaxOPF}$_{\sf V}(I')$ gives a total utility at least $\alpha$.


Conversely, assume that the $(\alpha,\beta)$-approximation algorithm gives a solution $x\in\{0,1\}^{m+1}$ of total utility at least $\alpha$. Customer $m+1$ must be satisfied in this solution. Then, Cons.~\raf{eq:betac2} implies
\begin{align*}
& \frac{v_0 - \beta v_{\max}}{2} \le \sum_{k=1}^m   \overline{\Lambda}(a_k x_k -  B) +\ell_e \le  \frac{v_0 - \tfrac{1}{\beta} v_{\min}}{2} \notag \\
\Rightarrow & - \frac{ (\beta v_{\max}-v_0) }{2 \overline{\Lambda}} -\frac{\ell_e}{\overline{\Lambda}} \le \sum_{k=1}^m a_k x_k -  B \le \frac{ v_0 - \tfrac{1}{\beta} v_{\min} }{2\overline{\Lambda}} - \frac{\ell_e}{\overline{\Lambda}} \label{eq:h2}
\end{align*}

The right-hand side can be bounded by
$$ \frac{v_0 - \tfrac{1}{\beta} v_{\min}}{2 \overline{\Lambda}}  - \frac{\ell_e}{\overline{\Lambda}}  \le \frac{1}{2}  < 1$$
Using Eqn.~\raf{eq:le}, the left-hand side can be bounded by
$$  - \frac{ (\beta v_{\max}-v_0 )}{2 \overline{\Lambda}} -\frac{\ell_e}{\overline{\Lambda}}  \ge -\frac{1}{2} - \frac{\tfrac{\delta}{2}- \epsilon}{ \overline{\Lambda}}> -1$$

Since $|\sum_{k=1}^m a_k x_k - B| < 1$,	 and $a_k, B$ are integers, this implies  $ \sum_{k=1}^m a_k x_k - B = 0$. Hence, \textsc{SubSum}$(I)$ is a ``yes" instance.
\end{proof}

\begin{customthm}{2} 
	Unless P=NP, there exists no ($\alpha$, $\beta$)-approximation for \textsc{MaxOPF}$_{\sf C}$ in general networks, for any $\alpha$ and $\beta$ have polynomial length in $n$, even in purely resistive electric networks (i.e., ${\rm Im}(z_{i,j})= 0$ for all $(i, j) \in \cE$ and ${\rm Im}(s_k) = 0$ for all $k \in {\cal N}$).
\end{customthm}

\begin{proof}
The basic idea is similar to that of Theorem~\ref{thm:OPFV-hard}. We consider a purely resistive electric network that contains a cycle.

Given a \textsc{SubSum} instance $I = (A, B)$, where $A = \{a_1,...,a_m \}$, we define a $\textsc{MaxOPF}_{\sf C}$ instance $I'$ as follows.
Define the customers' demands and utilities by
$$
s_k  \triangleq  a_k, \quad   u_k \triangleq  \tfrac{\alpha}{m+1}	 \quad
s_{m+1} \triangleq   B, \quad  u_{m+1} \triangleq 1
$$
Consider the network in Fig.~\ref{f5} for \textsc{MaxOPF}$_{\sf C}$ such that all power demands $\{ s_k \}_{k=1,...,m}$ are attached to node $a$,  $s_{m+1}$ is attached to $b$, and $z_{0, a} = z_{0, b} =  z_{a, b} = 1$.

\begin{figure}[!htb]
\begin{center}
	\includegraphics{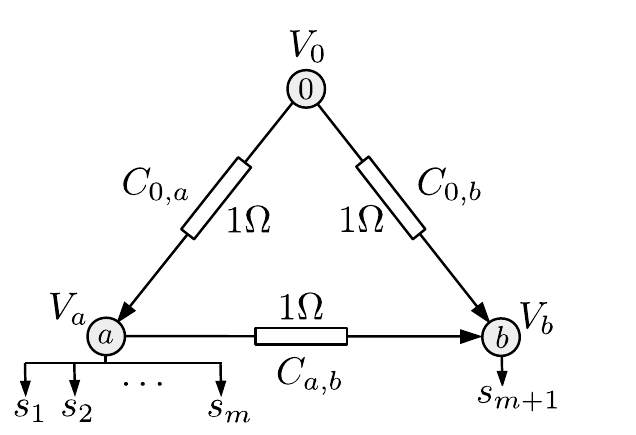}
\end{center}
\caption{A gadget for reduction from \textsc{SubSum} to \textsc{MaxOPF}$_{\sf C}$.}
\label{f5}
\end{figure}

Denote the transmitted power, current, and resistance on edge $(i,j)$ by $S_{i,j}$, $I_{i,j}$, $z_{i,j}$ respectively. Let $S_a\triangleq \sum_{k=1}^{m} s_k x_k$ and  $S_b\triangleq s_{m+1} x_{m+1}$ be the total demand on node $a$ and $b$ respectively.
Without loss of generality, assume $V_a \ge V_b$. By the power balance equations, we obtain:
\begin{align*}
S_a &= S_{0,a}- I^2_{0,a} z_{0,a} - S_{a,b},\\
S_b &= S_{0,b} - I^2_{0,b} z_{0,b} + S_{a,b} - I^2_{a,b}z_{z,b}.
\end{align*}
Using Ohm's law $I_{i,j} = \frac{V_i - V_j}{z_{i,j}}$ and $S_{i,j} = V_i I_{i,j}$, we obtain:
\begin{align*}
S_a &= \frac{V_0 (V_0 - V_a)}{z_{0,a}} - \frac{(V_0 - V_a)^2}{z_{0,a}} - \frac{V_a (V_ a - V_b)}{z_{a,b}},\\
S_b &= \frac{V_0 (V_0 - V_b)}{z_{0,b}} - \frac{(V_0 - V_b)^2}{z_{0,b}} + \frac{V_a (V_ a - V_b)}{z_{a,b}} -  \frac{(V_ a - V_b)^2}{z_{a,b}}
\end{align*}
Note that the above equations can also obtained when $V_b \ge V_a$.

Since $z_{0,a}=z_{0,b} = z_{a,b} = 1$, we obtain:
\begin{align*}
2 V_a^2 -  (V_0 + V_b) V_a + S_a = 0 \\
2 V_b^2 -  (V_0 + V_a) V_b + S_b = 0
\end{align*}
It follows that
\begin{align}
S_a - S_b & = 2 V_b^2 - 2 V_a^2 + V_0 (V_a - V_b) \\
& = (V_b - V_a) (2 V_b + 2 V_a - V_0) \label{eqn:pwr-diff} 
\end{align}

Let $x=(x_1,\ldots,x_{m+1})$ be a solution of the $(\alpha,\beta)$-approximation to \textsc{MaxOPF}$_{\sf C}$, where $x_k$ indicates if power demand $s_k$ is satisfied for $k \in \{1, ..., m\}$, and $x_{m+1}$ indicates if power demand $s_{m+1}$ is satisfied.  If the total utility of $x$ is at least $\alpha$, then we necessarily have $x_{m+1}=1$. 
Considering the capacity constraints, we obtain:
\begin{align*}
|V_a (V_a - V_b)| & \le \beta C_{a,b}, \\
|V_0 (V_0 - V_a)| & \le \beta C_{0,a}, \\  
|V_0 (V_0 - V_b)| & \le \beta C_{0,b}
\end{align*}
Note that $V_0 > V_a$, because $V_0$ is attached to generation. Then, we obtain:
\begin{equation}
\frac{V_0^2  - \beta C_{0,a}}{V_0} \le  V_a \label{eqn:pwr-diff2}
\end{equation}

By substituting Eqn.~(\ref{eqn:pwr-diff}), Eqn.~(\ref{eqn:pwr-diff2}) and considering $V_0 > V_a$ and $V_0 > V_b$, we obtain:
\begin{align*}
|V_a (V_a - V_b)| & \le \beta C_{a,b} \\
 \Rightarrow \qquad |S_a - S_b| & \le  \beta C_{a,b} \Big|\frac{2 V_b + 2 V_a - V_0}{V_a} \Big| \\
 & \le \beta C_{a,b} \frac{3 V_0}{V_a} \le \frac{3 \beta C_{a,b} V_0^2}{V_0^2  - \beta C_{0,a}}
\end{align*}

Next, we set $C_{a,b} < \frac{V_0^2  - \beta C_{0,a}}{3 \beta V_0^2}$, such that 
$$
|V_a (V_a - V_b)| \le \beta C_{a,b} \; \Rightarrow \; |S_a - S_b| = \bigg| B - \sum_{k = 1}^m a_k x_{k}\bigg| < 1
$$
Thus, \textsc{SubSum}$(I)$ is a ``yes" instance.

Conversely, a feasible solution $x\in\{0,1\}^{m+1}$ satisfying $\sum_{k=1}^na_kx_k-Bx_{m+1}=0$, with $x_{m+1}=1$, we can see that $S_a = S_b = B$. Next, we set $V_{a} = V_{b} = V'$ for some positive value $V' < V_0$ and $C_{0,a} = C_{0,b} = B + (V_0 - V')^2$. This is a feasible solution (with $\beta=1$) to $\textsc{MaxOPF}_{\sf C}(I')$, which has utility at least 1. Thus the $(\alpha,\beta)$-approximation returns a solution of utility at least $\alpha$.
\end{proof}

\begin{customlem}{5} 
	Define $W_j^{(1)}$ as in Eqn.~\raf{eq:v1}. We obtain:
	\begin{equation}
	|W_j^{(1)}| \le \lfloor \sec \theta \cdot \sec\tfrac{\theta}{2}  \rfloor+ 1 \label{eq:b3}
	\end{equation}
\end{customlem}
\begin{proof}
Assume $|W_j^{(1)}| > 1$. We note that based on the tree topology, all demand paths share a single source (i.e., the root).
When adding demand $s_{k_j}$ to $R_{j-1}$, Eqn.~\raf{eq:v1} implies that each element of $W_j^{(1)}$ if added to $R_j$  must cause violation at some (possibly more than one) edges. These violations occur only along the path $P_{k_j}$. Denote by $E \subseteq P_{k_j}$  the set of edges at which violations occur (after adding some $k \in W_j^{(1)}$ to $R_j$). Define $e^\circ \in E$ to be the closest edge to the root satisfies $e^\circ \in P_k$ for all  $k \in W_j^{(1)}$ because all demands share the same source (see Fig.~\ref{fig:treepf}(a)). This property allows us to bound $| W_j^{(1)}|$.

\begin{figure*}[!htb]
	\begin{subfigure}[b]{0.5\textwidth}
		\begin{center}
			\includegraphics[scale=.6]{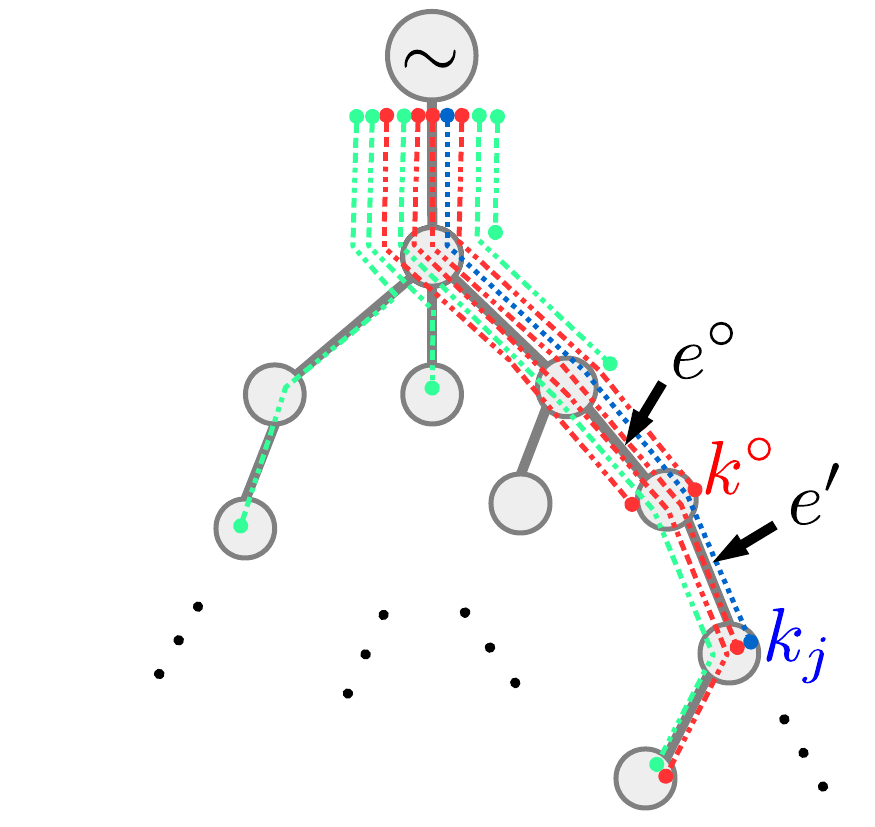}
		\end{center}
		\caption{}
	\end{subfigure}
	\begin{subfigure}[b]{0.5\textwidth}
		\begin{center}
			\includegraphics[scale=0.6]{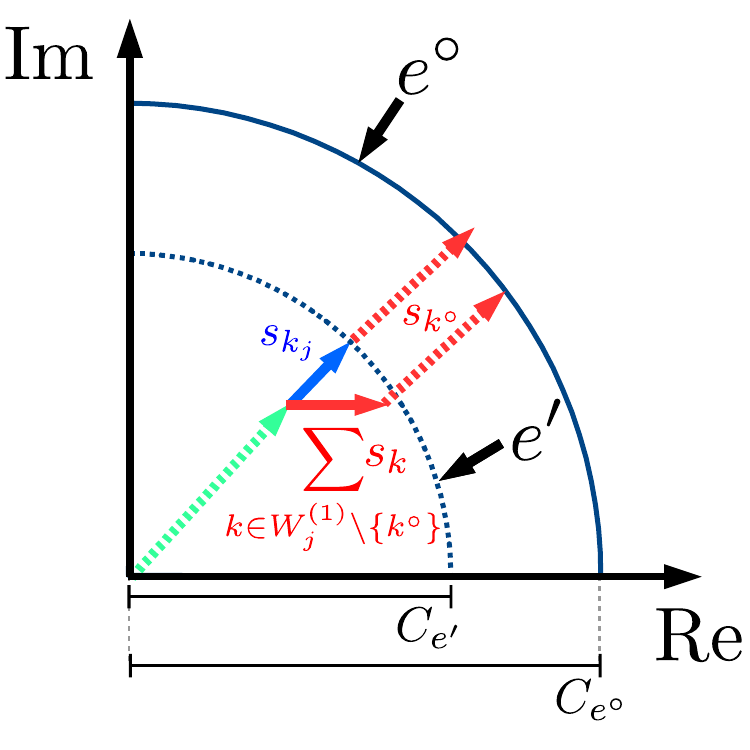}
		\end{center}
		\caption{}
	\end{subfigure}
	\caption{(a) The dotted lines illustrate the demand paths in $R_{j}\backslash \{k_j\}$ (green), $W_j^{(1)}$ (red), and $k_j$ (black). The paths of  $W_j^{(1)}$ and $k_j$ intersect at edge ${e^\circ}$. Adding any demand from $W_j^{(1)}$ (red) to $R_j$  (green and black) will violate the pwoer capacity constraints at $C_{{e^\circ}}$ and $C_{e'}$. (b) The horizontal and vertical axes correspond to the real and imaginary components of complex-valued demands respectively. The circular lines visualize Cons.~\raf{con} at edges ${e^\circ}$ and $e'$ respectively. The dotted green arrow corresponds to $\sum_{k \in R_{j}\backslash  \{k_j\}} s_k$, the solid black is $s_{k_j}$, the solid red is $\sum_{k \in W_j^{(1)}\backslash \{k^\circ\}} s_{k}$, and the dotted red arrows is $s_{k^\circ}$ (which is replicated to illustrate the violation). }
	\label{fig:treepf}
\end{figure*}  

More specifically, define customer $k^\circ \in W_j^{(1)}$, such that 
$$\Bigg|\sum_{k: k \in R_j, e^\circ \in P_k} s_k +  s_{k^\circ } \Bigg| > C_{e^\circ}$$
See Fig.~\ref{fig:treepf}(b) for an illustration. Note that
\begin{equation}
\Bigg|\sum_{k: k \in R_j \wedge e^\circ \in P_k} s_k+  s_{k^\circ } \Bigg| \ > \ C_{e^\circ} \ \ge \ \Bigg|\sum_{\substack{k: k \in  R_j\cup W_j^{(1)}\backslash \{k_j\} \\ \wedge e^\circ \in P_k }} s_k \Bigg| \label{eq:rel}
\end{equation}

\begin{claim} \label{claim:kj}
Given $k_j \in A_i$ for some $i$, then $W_j \cap \bigcup_{l=1}^{i-1} B_{l} = \varnothing$. Thus, for any $k_j \in M $, we have:
\begin{equation}
|s_{k_j}|\le | s_{k}| \mbox{ \ for all \ } k \in W_j \label{eq:order}
\end{equation}
\end{claim}
We prove Claim~\ref{claim:kj} as follows. First, since $R_{j-1} \supseteq \cup_{l=1}^{i-1} A_{l}$, it follows that $R_{j-1} \cap B_{l} = \varnothing$ for all $l \le i-1$, because $R_{j-1}$ is a feasible solution. Therefore, $W_j \subseteq \bigcup_{l=i}^m B_{l}$ and $W_j \cap \bigcup_{l=1}^{i-1} B_{l} = \varnothing$. Then, Eqn.~\raf{eq:order} follows from the non-decreasing order of demands $|s_1| \le |s_2| \le ....$
 
Hence, we obtain:
\begin{equation} 
{|s_{k_j}|} \le \frac{\sum_{k \in W_j^{(1)}\backslash\{k^\circ\}} |s_k|}{|W_j^{(1)}|-1} \label{eq:rel2}
\end{equation}

Rearranging Eqn.~\raf{eq:rel2} and using the fact that $|W_j^{(1)}|$ is an integer, we apply Lemma~\ref{lem:tb} 
\ifsupplementary
(in the appendix)
\fi 
to obtain:
\begin{align*}
|W_j^{(1)}|  &\le \Bigg\lfloor \frac{ \sum_{k \in W_j^{(1)}\backslash \{k^\circ\}}|s_k| }{|s_{k_j}| } + 1 \Bigg\rfloor \\
&\le \Bigg\lfloor  \frac{\sec \tfrac{\theta}{2} \cdot \Big|\sum_{k \in W_j^{(1)}\backslash \{k^\circ\} } s_k \Big| }{|s_{k_j}|}  \Bigg\rfloor +1
\end{align*}
Let 
$$
d_0 = \sum_{k \in R_j \backslash \{k_j\},e^\circ \in P_k } s_k+s_{k^\circ}, \ \
d_1 = s_{k_j}, \ \ d_2 = \sum_{k \in W_j^{(1)} \backslash \{k^\circ\}}  s_k
$$
By Eqn.~\raf{eq:rel}, it follows that $|d_0+d_1|  > |d_0+d_2|$.

Next, we apply Lemma~\ref{lem:ratio}
\ifsupplementary
 (in the appendix) 
\fi
to bound $\big|\sum_{k \in W_j^{(1)}\backslash \{k^\circ\} } s_k \big|/|s_{k_j}| \le  \sec \theta$, and obtain:
$$
|W_j^{(1)}| \le \lfloor \sec \theta \cdot \sec\tfrac{\theta}{2}  \rfloor+ 1 
$$
\end{proof}

\begin{customlem}{6} 
	Define $W_j^{(2)}$ as in Eqn.~\raf{eq:v2}. We obtain:
	\begin{equation}
	|W_j^{(2)}| \le \lfloor \eta\cdot {\rho\cdot\sec\theta_{\rm zs}}\rfloor + 1 \label{eq:vpf3}.
	\end{equation}
\end{customlem}
\begin{proof}
Assume $|W_j^{(2)}|>1$. Let $k' \in W_j^{(2)} $ be an arbitrary customer. Considering Cons.~\raf{conV2}, define edge ${\tilde{e}} \in \cL$ such that
\begin{equation}
	Q_{k_j}({\tilde{e}}) \ge \sum_{k\in W_j^{(2)}} Q_k({\tilde{e}}) -  Q_{k'}({\tilde{e}}).\label{eq:vpf}
\end{equation}
Note that ${\tilde{e}}$ must exist, otherwise
$$Q_{k_j}( e) +  Q_{k'}( e) < \sum_{k\in W_j} Q_k( e),  \forall e \in \cL$$
This implies that $W_j^{(2)}$ is not minimal, namely, $\{k'\} \cup R_j$ is a feasible solution, which contradicts the definition of $W_j^{(2)}$ in  Eqn.~\raf{eq:v2}. Let {$e_{\max} \triangleq  \arg\max_{e \in P_{k_j}}  z_e^{\rm R} s_{k_j}^{\rm R} + z_e^{\rm I} s_{k_j}^{\rm I}$}.

Then, for all $k \in W_j^{(2)}\backslash\{k'\}$ and $e \in P_{k}\cap P_{{\tilde{e}}}$, we obtain:
\begin{align}
z_{e_{\max}}^{\rm R} s_{k_j}^{\rm R} + z_{e_{\max}}^{\rm I} s_{k_j}^{\rm I} = \ &  \frac{z_{e_{\max}}^{\rm R} s_{k_j}^{\rm R} + z_{e_{\max}}^{\rm I} s_{k_j}^{\rm I}}{z_{e}^{\rm R} s_{k}^{\rm R} + z_{e}^{\rm I} s_{k}^{\rm I}} (z_{e}^{\rm R} s_{k}^{\rm R} + z_{e}^{\rm I} s_{k}^{\rm I}) \notag\\
\le  \ & \frac{|z_{e_{\max}}|\cdot|s_{k_j}| }{|z_e|\cdot|s_{k}|\cdot \cos \theta_{\rm zs}} (z_{e}^{\rm R} s_{k}^{\rm R} + z_{e}^{\rm I} s_{k }^{\rm I}) \label{eq:w1z}\\
\le  \ & \tfrac{\rho}{\cos\theta_{\rm zs}}\cdot (	z_{e}^{\rm R} s_{k}^{\rm R} + z_{e}^{\rm I} s_{k}^{\rm I}), \label{eq:w2z}
\end{align}
where Eqn.~\raf{eq:w1z} follows by Cauchy-Schwarz inequality and $\theta_{\rm zs} \in [0,\tfrac{\pi}{2})$, Eqn.~\raf{eq:w2z}  by $\frac{|s_{k_j}|}{|s_k|}\le 1$ by Claim~\ref{claim:kj}. 

Summing all $k \in W_j^{(2)}\backslash\{k'\}$ and $e \in P_{k}\cap P_{{\tilde{e}}}$, we obtain:
\begin{align}
& z_{e_{\max}}^{\rm R} s_{k_j}^{\rm R} + z_{e_{\max}}^{\rm I} s_{k_j}^{\rm I}\notag\\
\le \  &  \bigg(\frac{\rho}{\cos\theta_{\rm zs}} \bigg)\cdot \frac{  \sum_{k \in W_j^{(2)}\backslash\{k'\}} \sum_{e \in P_{k} \cap P_{{\tilde{e}}} } 	z_e^{\rm R} s_{k}^{\rm R} + z_e^{\rm I} s_{k}^{\rm I}}{\sum_{k \in W_j^{(2)} \backslash \{k'\}} |P_k \cap P_{{\tilde{e}}} |} \notag\\
\le  \ & \bigg(\frac{\rho}{\cos\theta_{\rm zs}} \bigg) \cdot \frac{ \sum_{k \in W_j^{(2)}\backslash\{k'\}} Q_k({\tilde{e}})}{|W_j^{(2)}| -1}  \label{eq:vpf1},
\end{align}
because of $|P_k \cap P_{{\tilde{e}}}| \ge 1$.
By the definition of $e_{\max}$, we obtain:
\begin{align}
  (z_{e_{\max}}^{\rm R} s_{k_j}^{\rm R} + z_{e_{\max}}^{\rm I} s_{k_j}^{\rm I}) \ge \ & \frac{1}{| P_{k_j} \cap P_{{\tilde{e}}}|}	\sum_{e \in P_{k_j} \cap P_{{\tilde{e}}} } z_{e}^{\rm R} s_{k_j}^{\rm R} + z_{e}^{\rm I} s_{k_j}^{\rm I}\notag \\
   \ge \ &  \frac{1}{\eta} \cdot	Q_{k_j}({\tilde{e}}) \label{eq:vpf2}
\end{align}
By Eqns.~\raf{eq:vpf}, \raf{eq:vpf1}, \raf{eq:vpf2} and $|W_j^{(2)}|$ as an integer, we obtain:
\begin{align*}
|W_j^{(2)}| &\le \left\lfloor\eta\cdot \bigg(    \frac{\rho}{\cos\theta_{\rm zs}} \bigg)  \frac{\sum_{k \in W_j^{(2)}\backslash\{k'\}}Q_k({\tilde{e}})}{Q_{k_j}({\tilde{e}})} \right\rfloor + 1 \notag \\
&\le  \lfloor \eta\cdot {\rho\cdot\sec\theta_{\rm zs}}\rfloor + 1 \label{eq:vpf3}
\end{align*}
\end{proof}

\begin{customlem}{7} \label{lem:ratio}
	Given vectors $d_0, d_1,d_2 \in \RR_+^2$ such that $|d_0 + d_1|  \ge |d_0+d_2|$, and $\arg(d_0),\arg(d_1), \arg(d_2) \in [0,\tfrac{\pi}{2})$; then
	$$ \frac{|d_2|}{|d_1|} \le \sec \theta,$$
	where $\theta$ is the maximum angle between any pair of vectors and $0 \le \theta < \frac{\pi}{2}$.
\end{customlem}
\begin{proof}
For $\theta = 0$, the statement is trivially true. Therefore, we assume otherwise.
Let $C \in \RR_+$ be a real number such that $|d_0 + d_1|  \ge C \ge |d_0+d_2|$.  Using the triangular inequality $|x|+|y| \ge |x+y|$ for $x,y \in \RR$ obtain:s, 
\begin{equation}
|d_1|  +|d_0|\ge |d_1 + d_0| \ge C \Rightarrow |d_1| \ge C- |d_0|. \label{eq:d1}
\end{equation}
Therefore,
$$
\frac{|d_2|}{|d_1|} \le \frac{|d_2|}{C-|d_0|} \le \frac{|d_2'|}{1-|d_0'|},
$$
where $d_2'$ and $d_0'$ are obtained from $d_2$ and $d_0$ respectively by dividing their magnitudes by $C$. We observe that  
$|d'_2 +d_0'| \le 1$.
We rewrite $|d'_2 + d_0'|$ as:
\begin{align*}
|d_2' +d_0'| & = 
|d_2'|^2 +|d_0'|^2 +  2 |d_2'| \cdot |d_0'| \cos (\theta_{d_2} - \theta_s)\\
& \le
|d_2'|^2 +|d_0'|^2 +  2 |d_2'| \cdot |d_0'| \cos (\theta) \le 1,
\end{align*}
where $\theta_{d_2}$ and $\theta_s$ are the arguments of $d_2'$ and $d_0'$ respectively.
We solve $|d_2'|^2 +|d_0'|^2 +  2 |d_2'| \cdot |d_0'| \cos (\theta) \le 1 $ for $|d_0'|$ to obtain:
$$
|d_0'| \le \sqrt{|d_2'|^2\cos^2 \theta - |d_2'|^2 + 1} - |d_2'| \cos \theta.
$$

Let $f(x) \triangleq \frac{\sqrt{x^2 \cos^2 \theta - x^2 + 1} - x \cos \theta}{1-x}$, for $x \in (0,1)$. Notice that $\frac{|d_0'|}{1-|d_2'|}  \le f(|d_2'|)$. We take the first derivative of $f(x)$
$$
f'(x) = \frac{-  \cos\theta  \sqrt{x^2   \cos^2 \theta -x^2+1}-x   \sin^2\theta +1}{(x-1)^2 \sqrt{x^2   \cos^2\theta-x^2+1}}
$$
Then, we solve $f'(x) = 0$ for $x$
\begin{align*}
&	-  \cos\theta  \sqrt{x^2   \cos^2 \theta -x^2+1}-x   \sin^2\theta +1 = 0\\
\Rightarrow \  &    \cos\theta  \sqrt{-x^2   \sin^2 \theta +1}+x   \sin^2\theta -1 = 0
\end{align*}
Dividing both sides by $\cos^2 \theta$,
$$
\Rightarrow \ \ \sqrt{-x^2   \tan^2 \theta +\frac{1}{\cos^2 \theta }}+x   \tan^2\theta -\frac{1}{\cos^2 \theta} = 0
$$

Rearranging the equation, squaring both sides, and using $\frac{1}{\cos^2 \theta} = 1 + \tan^2 \theta $, we obtain:
\begin{align*}
&- x^2   \tan^2 \theta + \tan^2 + 1\\
= \ &  (\tan^2 \theta+1)^2 - 2 x \tan^2\theta (\tan^2 \theta +1) + x^2 \tan^4\theta  \\
\Rightarrow \quad & (\tan^4 \theta + \tan^2 \theta ) (x^2 -2  x + 1) = 0 \\
\Rightarrow \quad & x = 1
\end{align*}
Since the critical point occurs on the boundaries, it is sufficient to check the values of $f(0)$ and $\lim_{x \to 1} f(x)$.
We find $\lim_{x \to 1} f(x)$ by applying L'Hospital's rule
\begin{align*}
&\lim_{x\to 1}	\frac{\frac{d}{d x} ( \sqrt{x^2 \cos^2 \theta - x^2 + 1} - x \cos \theta)}{\frac{d}{d x} (1-x) }\\
= \ &  \lim_{x \to 1} \frac{x-x\cos^2\theta }{\sqrt{x^2 \cos^2 \theta - x^2 + 1 }} + \cos \theta \\
= \ & \frac{1-\cos^2 \theta}{\cos \theta} + \cos \theta = \frac{1}{\cos \theta}
\end{align*}
Finally, we observe that $f(0)=1 \le  \lim_{x \to 1} f(x) = \frac{1}{\cos \theta}$, therefore the supremum is $\frac{1}{\cos \theta }$.
\end{proof}


\begin{customlem}{8} 
	\label{lem:tb}
	Given a set of 2D vectors $\{d_i \in \RR^2\}_{i=1}^n$
	$$ \frac{\sum_{i=1}^n |d_i| }{\bigg| \sum_{i =1}^n d_i \bigg|} \le {\sec \tfrac{\theta}{2}},$$
	where $\theta$ is the maximum angle between any pair of vectors and $0 \le \theta \le \frac{\pi}{2}$.
\end{customlem}
\begin{proof}
If $\theta=0$ then the statement is trivial, therefore we assume otherwise. We prove $\frac{(\sum_{i=1}^n |d_i| )^2}{|\sum_{i=1}^n d_i |^2} \le \frac{2}{\cos + 1}$  by induction (notice that $\sec \tfrac{\theta}{2} = \sqrt{\frac{2}{\cos \theta + 1}}$).  First, we expand  the left-hand side by
\begin{align}
&\frac{  \sum_{i=1}^n |d_i|^2 +  2\sum_{1\le i < j \le n} |d_i| \cdot |d_j|  } { \sum_{i =1}^n |d_i|^2 +  2\sum_{1\le i < j \le n} |d_i| \cdot |d_j| (\sin \theta_i \sin \theta_j + \cos \theta_i \cos \theta_j)}\notag\\
&=\frac{ \sum_{i=1}^n |d_i|^2 +  2\sum_{1\le i < j \le n} |d_i| \cdot |d_j|  } { \sum_{i=1}^n |d_i|^2 +  2\sum_{1\le i < j \le n} |d_i| \cdot |d_j| \cos (\theta_i - \theta_j)}, \label{eq:ind}
\end{align}
where $\theta_i$ is the angle that $d_i$ makes with the $x$ axis.

Consider the base case: $n= 2$. Eqn.~\raf{eq:ind} becomes
\begin{align}
\frac{|d_1|^2 +|d_2|^2 + 2 |d_1|\cdot |d_2|}{|d_1|^2 +|d_2|^2 + 2 |d_1|\cdot |d_2| \cos(\theta)} = f\Big(\frac{|d_2|}{|d_1|}\Big),
\end{align}
where $ f(x) \triangleq \frac{1+x^2+2x}{1+x^2+2x\cos \theta}$. The first derivative is given by
$$f'(x) = \frac{(1+x^2+2x \cos \theta )(2x + 2) - 1+x^2 + 2x)(2x + 2 \cos \theta)}{(1 + x^2 + 2x \cos \theta )^2}$$ 
$f'(x)$ is zero only when $x=1$. Hence, $f(1)$ is an extreminum point.  We compare $f(1)$ with $f(x)$ at the  boundaries $x\in \{0,\infty\}$:$$f(1) = \frac{2}{\cos \theta + 1} \ge f(0) = \lim_{x \to \infty} f(x) = 1$$
Therefore, $f(x)$ has a global maximum of $\frac{2}{\cos \theta  + 1}$.

Next, we proceed to the inductive step. We assume $\frac{\sum_{i=1}^{r-1} |d_i| }{\big| \sum_{i=1}^{r-1} d_i \big|} \le \sqrt{\frac{2}{\cos \theta + 1}}$ where $r \in \{1,\ldots, n\}$. W.l.o.g., assume $\theta_2\ge\theta_3\ge \cdots \ge \theta_n \ge \theta_1$. Rewrite Eqn.~\raf{eq:ind} as
\begin{equation}
\frac{ ( \sum_{i=1}^r |d_i|)^2 } { \displaystyle \sum_{i=1}^r |d_i|^2 + 2 \smashoperator{\sum_{1\le i<j<r}} |d_i|  |d_j| \cos (\theta_i - \theta_j) +  2 |d_r| \smashoperator{\sum_{1\le i<r}} |d_i|  \cos (\theta_i - \theta_r)} \label{eq:den}
\end{equation}
Let $g(\theta_r)$ be the denominator of Eqn.~\raf{eq:den}. We take the  second derivative of $g(\theta_r)$:
$$
g''(\theta_r) = -2 |d_r| \sum_{1 \le i < r} |d_i| \cos(\theta_i - \theta_r)
$$ 
Notice that $\cos(\theta_i - \theta_r) \ge 0$, therefore the second derivative is always negative. This indicates that all local exterma in $[0,\theta_{r-1}]$ of $g(\theta_n)$ are local maxima. Hence, the minimum occurs at the boundaries:
$$
\min_{\theta_r \in [0,\theta_{r-1}]}  g(\theta_r) \in \{g(0), g(\theta_{r-1})\}
$$
If $\theta_r \in \{0,\theta_r\}$ , then there must exist at least a pair of vectors in $\{d_i\}_{i=1}^r$ with the same angle. Combining these two vectors into one, we can obtain an instance with $r-1$ vectors. Hence, by the inductive hypothesis, the same bound holds up to $r$ vectors.
\end{proof}

\begin{customthm}{7}
	Assume that $\rho$, $\sec\theta$ and $\sec\theta_{zs}$ are constants, and $\theta, \theta_{zs} < \tfrac{\pi}{2}$, then		
	\begin{enumerate}
		\item {\sc InelasDemAlloc$_{\sf C}$} is  $\frac{1}{O( \log n)}$-approximation for {\sc sMaxOPF$_{\sf C}$}.
		\item {\sc InelasDemAlloc$_{\sf V}$} is $\frac{1}{O(\eta\cdot \log n)}$-approximation for {\sc sMaxOPF$_{\sf V}$}.
		\item {\sc InelasDemAlloc} is $\frac{1}{O(\eta\cdot \log n)}$-approximation for \textsc{sMaxOPF}.
	\end{enumerate}
\end{customthm}
\begin{proof}
By rounding utilities in {\sc InelasDemAlloc}, $\bar u_k \in \{0,...,n^2\} =  \{0,...,2^{2 \log n}\}$ for all $k\in \cN$. Therefore there are at most $2 \log n + 1$ groups of users (denoted by $\hat\cN_1,...,\hat\cN_{2\log n + 1}$ respectively). Let $M_1,...,M_{2\log n + 1}$ be their respective unit-utility solutions, returned by Algorithm~\ref{alg1}. Define $\OPT$ to be an optimal solution value for \textsc{sMaxOPF} (resp., \textsc{sMaxOPF}$_{\sf C}$, \textsc{sMaxOPF}$_{\sc V}$) and $R_i^\ast$, $i \in \{1,...,2\log n + 1\}$ be the subset of this optimal solution that belongs to group $i$. 
Clearly, $u_{\max}\le\OPT$, assuming each load can be individually served (those loads that cannot be individually served can be determined by checking the feasibility of the problem with exactly one load turned on). 
Define $u(N)\triangleq \sum_{k \in N} u_k$ for any $N \subseteq \cN$. The solution returned by {\sc InelasDemAlloc} satisfies the following:
\begin{equation}
u(M)=\max_{i \in \{1,...,2\log n +1\}} u(M_i) \ge \frac{\sum_{i=1}^{2\log n+1} u(M_i)}{2 \log n + 1}    \label{eq:xl0}
\end{equation}

Using the fact that $x \ge y \big\lfloor \tfrac{x}{y} \big\rfloor$ for any $x,y \in \RR$ and $y \neq 0$, we obtain for $i\in\{1,...,2\log n + 1\}$,
\begin{align}
u(M_i) = \sum_{k \in M_i} u_k \ge L \sum_{k \in M_i} \bar u_k \ge \alpha L \sum_{k \in R^\ast_i} \bar u _k\label{eq:xl1} 
\end{align}
where $\alpha$ is the approximation ratio of {\sc GreedyAlloc} on unit-utility instances.

Next, we use $x \le y \big\lfloor \tfrac{x}{y} \big\rfloor + y$ for any $x,y \in \RR$ and $y \neq 0$ to obtain:
\begin{align}
&\sum_{k \in R^\ast_i}  u_k \le 	\sum_{k \in R^\ast_i} (L \bar u_k + L) = 	L \sum_{k \in R^\ast_i} \bar u_k  + |R^\ast_i| L  \notag \\
\Rightarrow  \quad &	L \sum_{k \in R^\ast_i} \bar u_k  \ge \sum_{k \in R^\ast_i}  u_k  - |R^\ast_i| L \label{eq:xl2}
\end{align}

Finally, we complete the proof using Eqns.~\raf{eq:xl0}-\raf{eq:xl2}:	
\begin{align*}
u(M) \ge \ & \frac{\alpha}{2 \log n + 1} \sum_{i=1}^{2\log n+1} (\sum_{k \in R^\ast_i} u_k - |R^\ast_i| L )\\
= \ &  \frac{\alpha}{2 \log n + 1} ( \OPT - \sum_{i=1}^{2\log n+1} |R^\ast_i| L)\\
\ge \ &  \frac{\alpha}{2 \log n + 1} ( \OPT - \frac{n \cdot u_{\max }}{n^2} ) \\
\ge \ & \frac{\alpha}{2 \log n + 1}  \cdot  (1-\tfrac{1}{n}) \cdot \OPT
\end{align*}

Therefore,
\begin{enumerate}

\item {\sc InelasDemAlloc$_{\sf C}$} is  $\bar \alpha$-approximation for {\sc sMax-OPF$_{\sf C}$}, where
$$\bar \alpha = \frac{\Big(  \big\lfloor   \sec \theta \cdot \sec\tfrac{\theta}{2}    \big\rfloor + 1 \Big)^{-1}}{  2\log n    + 1}\cdot  \big(1-\tfrac{1}{n}\big) = \frac{1}{O( \log n)}$$

\item {\sc InelasDemAlloc$_{\sf V}$} is  $\bar \alpha$-approximation for {\sc sMax-OPF$_{\sf V}$}, where
$$\bar \alpha = \frac{\Big(\big\lfloor   \eta\cdot  \rho \cdot\sec\theta_{\rm zs} \big\rfloor + 1  \Big)^{-1}}{  2\log n    + 1} \cdot  \big(1-\tfrac{1}{n}\big) = \frac{1}{O(\eta\cdot \log n)}$$

\item {\sc InelasDemAlloc} is $\bar \alpha$-approximation for \textsc{sMaxOPF}, where
\begin{align*}
\bar \alpha & = \frac{\Big(\big\lfloor   \eta\cdot  \rho \cdot\sec\theta_{\rm zs} \big\rfloor + \big\lfloor   \sec \theta \cdot \sec\tfrac{\theta}{2}\big\rfloor + 2 \Big)^{-1}}{  2\log n     + 1} \cdot \big(1-\tfrac{1}{n}\big) \\
& = \frac{1}{O(\eta\cdot \log n)}
\end{align*}

\end{enumerate}

\end{proof}

\subsection{Hardness of {\sc sMaxOPF$_{\sf V}$}}

\begin{definition}
For $\alpha\in(0,1]$ and $\beta\ge1$, we define a bi-criteria $(\alpha,\beta)$-approximation to \textsc{sMaxOPF} as a solution $\hat x =\big((\hat{x}_k)_{k \in \cI}, (\hat x_k)_{k\in\cF}\big) \in \{0, 1\}^{|\cI|}\times[0,1]^{|\cF|}$ satisfying
\begin{align}
&\bigg| \sum_{k: e  \in P_k} s_k x_k \bigg| \le \beta \hat{C}_{e}, \qquad \forall  e \in \cE\\
&\frac{1}{\beta} \underline{V_e} \le \sum_{k \in \cN} \Big( \sum_{ e'  \in P_k \cap P_e} z^{\rm R}_{e'} s_k^{\rm R}+ z^{\rm I}_{e'} s_k^{\rm I} \Big) x_k \le \beta \overline{V_e}, \ \forall  e \in \cE 
\end{align}
such that $u(\hat x) \ge \alpha \OPT$.
\end{definition}

\begin{customthm}{8}\label{thm:sOPFV-hard}
Unless {\rm P}={\rm NP}, there is no $(\alpha,\beta)$-approximation for {\sc sMaxOPF$_{\sf V}$} (even when $|\cE| = 1$) by a polynomial-time algorithm in $n$, for any $\alpha$ and $\beta$ have polynomial length in $n$.
\end{customthm}

\begin{proof}
We present a reduction from \textsc{SubSum} to {\sc sMaxOPF$_{\sf V}$}.
Assume that there is an ($\alpha, \beta$)-approximation for \textsc{sMaxOPF}$_{\sf V}$. We construct an instance $I'$ of  \textsc{sMaxOPF}$_{\sf V}$ for each instance  $I$ of \textsc{SubSum}, such that \textsc{SubSum}$(I)$ is a ``yes'' instance if and only if the $(\alpha, \beta)$-approximation of 	\textsc{sMaxOPF}$_{\sf V}(I')$ gives a total utility at least $\alpha$.

We define the \textsc{sMaxOPF}$_{\sf V}$ instance $I'$ as follows. Consider a graph with a single edge $e$. Let $z_e \triangleq 1 + \bf i$, $\underline{V} = \underline{V_e}$ and $\overline{V} = \overline{V_e}$. Define  $\Lambda \triangleq \max \{-\tfrac{1}{\beta} \underline V, \beta \overline V\}$. Let  $\cN= \cI =  \{1,...,m+1\}$ be the set of customers (i.e., all having inelastic demands). For each $k \in \{1,...,m+1\}$, define the customers' demands and utilities by:
$$
s_k \triangleq 2  \Lambda a_k, \quad  u_k  \triangleq  \tfrac{\alpha}{m+1}	 \quad
s_{m+1} \triangleq - {\bf i} 2  \Lambda B, \quad   u_{m+1} \triangleq 1
$$

First, we prove that if \textsc{SubSum}$(I)$ is a ``yes'' instance, then the $(\alpha, \beta)$-approximation of \textsc{sMaxOPF}$_{\sf V}(I')$ gives a total utility at least $\alpha$. If \textsc{SubSum}$(I)$ is a ``yes'' instance, then $\sum_{k = 1}^{m}a_k \hat{x}_k = B$, where $\hat x \in \{0,1\}^m$ is a solution vector of \textsc{SubSum}. Construct a solution $x\in \{0,1\}^{m+1}$ of \textsc{sMaxOPF}$_{\sf V}$ such that
$$
x_k = \left\{
\begin{array}{l l}
\hat x_k & \quad \text{if $k = 1, ..., m$}\\
1 & \quad \text{if $k = m+1$}
\end{array} \right.
$$

\begin{table*}[!htb]
	\def\arraystretch{1.5}
	\centering
	{\scriptsize
		\begin{tabular}{|c|c|c|c||c|c|c|c||c|c|c|c|}
			\hline\multicolumn{12}{|c|}{Network Data}\\
			\hline 
			            &  (p.u.)  &   (p.u.)  & (p.u.)  &     &   (p.u.)  &  (p.u.)  & (p.u.) &     &   (p.u.)  &  (p.u.)  & (p.u.) \\
			\hline   $(0,2)  $ &  $0.000574$  & $0.000293 $ & $ 4.6$&    $(14,15)$ &  $0.00368 $  & $0.003275 $ & $ 0.3$&     $(27,28)$ &  $0.006594$  & $0.005814$ & $ 1.5$ \\
			\hline   $(2,3)  $ &  $0.00307 $  & $0.001564 $ & $ 4.1$&    $(15,16)$ &  $0.004647$  & $0.003394 $ & $ 0.25$&    $(28,29)$ &  $0.005007$  & $0.004362$ & $ 1.5$ \\
			\hline   $(3,4)  $ &  $0.002279$  & $0.001161 $ & $ 2.9$&    $(16,17)$ &  $0.008026$  & $0.010716 $ & $ 0.25$&    $(29,30)$ &  $0.00316 $  & $0.00161 $ & $ 1.5$ \\
			\hline   $(4,5)  $ &  $0.002373$  & $0.001209 $ & $ 2.9$&    $(17,18)$ &  $0.004558$  & $0.003574 $ & $ 0.1$&     $(30,31)$ &  $0.006067$  & $0.005996 $ & $ 0.5$ \\
			\hline   $(5,6)  $ &  $0.0051  $  & $0.004402 $ & $ 2.9$&    $(2,19) $ &  $0.001021$  & $0.000974 $ & $ 0.5$&     $(31,32)$ &  $0.001933$  & $0.002253 $ & $ 0.5$ \\
			\hline   $(6,7)  $ &  $0.001166$  & $0.003853 $ & $ 1.5$&    $(19,20)$ &  $0.009366$  & $0.00844  $ & $ 0.5$&     $(32,33)$ &  $0.002123$  & $0.003301 $ & $ 0.1$ \\
			\hline   $(7,8)  $ &  $0.00443 $  & $0.001464 $ & $ 1.05$&   $(20,21)$ &  $0.00255 $  & $0.002979 $ & $ 0.21$&    $(8,34) $ &  $0.012453$  & $0.012453 $ & $ 0.5$ \\
			\hline   $(8,9)  $ &  $0.006413$  & $0.004608 $ & $ 1.05$&   $(21,22)$ &  $0.004414$  & $0.005836 $ & $ 0.11$&    $(9,35) $ &  $0.012453$  & $0.012453 $ & $ 0.5$ \\
			\hline   $(9,10) $ &  $0.006501$  & $0.004608 $ & $ 1.05$&   $(3,23) $ &  $0.002809$  & $0.00192  $ & $ 1.05$&    $(12,36)$ &  $0.012453$  & $0.012453 $ & $ 0.5$ \\
			\hline   $(10,11)$ &  $0.001224$  & $0.000405 $ & $ 1.05$&   $(23,24)$ &  $0.005592$  & $0.004415 $ & $ 1.05$&    $(18,37)$ &  $0.003113$  & $0.003113 $ & $ 0.5$ \\
			\hline   $(11,12)$ &  $0.002331$  & $0.000771 $ & $ 1.05$&   $(24,25)$ &  $0.005579$  & $0.004366 $ & $ 0.5$&     $(25,38)$ &  $0.003113$  & $0.003113 $ & $ 0.1$ \\
			\hline   $(12,13)$ &  $0.009141$  & $0.007192 $ & $ 0.5$&    $(6,26) $ &  $0.001264$  & $0.000644 $ & $ 1.5$& \multicolumn{4}{c|}{\multirow{2}{*}{$S_{\text{base}}=1$MV, $V_{\text{base}}=12.66$KV}}  \\
			\cline{1-8}   $(13,14)$ &  $0.003372$  & $0.004439 $ & $ 0.45$&   $(26,27)$ &  $0.00177 $  & $0.000901 $ & $ 1.5$& \multicolumn{4}{c|}{}  \\
			\hline
		\end{tabular} 
	}
	\caption{The settings of line impedance and maximum capacity of the 38-node electric network.}
	\label{tab:net}
\end{table*}

By $\sum_{k = 1}^{m}a_k \hat x_k - B = 0$,  we obtain:
\begin{align*}
\sum_{k=1}^{m+1} (z_e^{\rm R}s^{\rm R}_k + z_e^{\rm I}s^{\rm I}_k)x_k
& = \sum_{k=1}^m 2 \Lambda a_k x_k - 2  \Lambda B x_{m+1} \\
&= \sum_{k=1}^m  2 \Lambda(a_k x_k -  B x_{m+1}) = 0
\end{align*}

Therefore, $(x_k)_{k\in \cN}$ is a feasible solution of {\sc sMaxOPF$_{\sf V}$} and satisfies Eqn.~\raf{conV}.
Since $u_{m+1}=1$, $u(x)$ is at least one which implies that $\OPT$ is also at least 1, and hence, by the feasibility of this solution, any $(\alpha, \beta)$-approximation gives a total utility at least $\alpha$. 

\medskip

Conversely, assume the $(\alpha,\beta)$-approximation gives a solution $x\in\{0,1\}^{m+1}$ of total utility at least $\alpha$. Since customer $m+1$ has valuation $v_{m+1}=1$, while the rest of customers valuations total to less than $\alpha$ (i.e., $\sum_{k = 1}^{m} u_k < \alpha$), customer $m+1$ must be satisfied in this solution.
Therefore, we obtain:
\begin{align}
& \displaystyle \frac{1}{\beta} \underline{V} \le \sum_{k=1}^m  2 \Lambda(a_k x_k - B) \le \beta \overline{V}\\
\Rightarrow \quad & \tfrac{\underline V}{2 \beta \Lambda} \le \sum_{k=1}^m a_k x_k -  B \le \tfrac{\beta \overline V}{2 \Lambda}
\end{align}

Since $-\tfrac{\underline V}{2 \beta \Lambda}, \tfrac{\beta \overline V}{2 \Lambda} \le \frac{1}{2} < 1$, and $a_k, B$ are integers, this implies  $ \sum_{k=1}^m a_k x_k - B = 0$. Hence, \textsc{SubSum}$(I)$ is a ``yes" instance.
\end{proof}

\subsection{Setting of Simulation}

We consider a 38-node electric network in Fig.~\ref{fig:net}. The settings of line impedance and maximum capacity are provided in Table~\ref{tab:net}. 

%

\end{document}